\documentclass[11pt,a4paper]{article}
\pdfoutput=1
\usepackage{jheppub}
\usepackage{graphicx,epsf,amssymb,amsbsy,amsfonts,amssymb,amsmath,amsthm,mathtools,subfigure}
\usepackage{hyperref}
\usepackage{mathrsfs}
\usepackage[usenames]{xcolor}

\newcommand{\eq}{\begin{equation}}
\newcommand{\eqe}{\end{equation}}
\newcommand{\g}{\gamma}
\newcommand{\e}{\epsilon}
\newcommand{\G}{\Gamma}
\newcommand{\La}{\Lambda}
\newcommand{\up}{\upsilon}

\newcommand{\eqa}{\begin{eqnarray}}
\newcommand{\eqae}{\end{eqnarray}}

\newcommand{\BSM}{\left(\begin{smallmatrix}}
\newcommand{\ESM}{\end{smallmatrix}\right)}
\newcommand{\BM}{\left(\begin{matrix}}
\newcommand{\EM}{\end{matrix}\right)}

\newcommand{\A}{\alpha} 
 
\newcommand{\Fu}{{\cal F}}
 
\newcommand{\s}{\sigma}

\newtheorem{thm}{Theorem}[section]
\newtheorem{cor}[thm]{Corollary}
\newtheorem{lem}[thm]{Lemma}
\newtheorem{prop}[thm]{Proposition}

\newcommand\eea{\end{eqnarray}}
\newcommand\bea{\begin{eqnarray}}

\def\R{\mathbb{R}}
\def\C{\mathbb{C}}

\def\Z{\mathbb{Z}}

\def\d{\partial}
\def\<{\langle}
\def\>{\rangle}
\def\+{\dagger}

\newcommand{\comment}[1]{}

\newcommand{\SL}{\operatorname{\textsl{SL}}}	 	%SL group
\newcommand{\GL}{\operatorname{\textsl{GL}}}	%GL group
\newcommand{\PSL}{\operatorname{\textsl{PSL}}}	%PSL group
\newcommand{\Li}{\operatorname{Li}}     			%Projection operator (for some reason) to kill polylogs???
\newcommand{\HH}{{\mathbb H}}				%Upper half-plane
\newcommand{\reg}{{\rm reg}}
\newcommand{\Reg}{{\rm Reg}}

\newcommand{\PF}{\widehat{P}} 				%regularization operator on F(\infty)
\newcommand{\RF}{\widehat{R}} 				%regularization operator on F(\infty)
\begin{document}

\title{L-functions for Meromorphic Modular Forms\\ and Sum Rules in Conformal Field Theory}

\author{David A. McGady}
\emailAdd{mcgady@nbi.ku.dk} 
\affiliation{The Niels Bohr International Academy, 17 Blegdamsvej\\ 
Copenhagen University, 2100 Copenhagen, Denmark}

\abstract{
We define L-functions for meromorphic modular forms that are regular at cusps, and use them to: (i) find new relationships between Hurwitz class numbers and traces of singular moduli, (ii) establish predictions from the physics of T-reflection, and (iii) express central charges in two-dimensional conformal field theories (2d CFT) as a literal sum over the states in the CFTs spectrum. When a modular form has an order-$p$ pole away from cusps, its $q$-series coefficients grow as $n^{p-1} e^{2 \pi n t}$ for $t \geq \tfrac{\sqrt{3}}{2}$. Its L-function must be regularized. We define such L-functions by a deformed Mellin transform. We study the L-functions of logarithmic derivatives of modular forms. L-functions of logarithmic derivatives of Borcherds products reveal a new relationship between Hurwitz class numbers and traces of singular moduli. If we can write 2d CFT path integrals as infinite products, our L-functions confirm T-reflection predictions and relate central charges to regularized sums over the states in a CFTs spectrum. Equating central charges, which are a proxy for the number of degrees of freedom in a theory, directly to a sum over states in these CFTs is new and relies on our regularization of such sums that generally exhibit exponential (Hagedorn) divergences. 
}

\maketitle

\section{L-functions, modular forms, and path integrals}\label{sec:intro}

In this paper, we define L-functions attached to meromorphic modular forms that are regular at cusps. To begin, consider an integer-weight modular form $f(\tau)$ defined on the upper half-plane $\HH$, and bounded away from cusps. Because $f(\tau) = f(\tau+1)$, it has a Fourier decomposition, $f(\tau) = \sum_n c_f(n) e^{2 \pi i n \tau}$. Taking the Mellin transform of $f$ yields the associated {\em L-function}~\cite{01-Apostol},
\begin{align}
L_f(s):= \sum_{n \neq 0} \frac{c_f(n)}{n^s}~.
\end{align}
This pairing between modular forms and Dirichlet series, $f(\tau) \leftrightarrow L_f(s)$, is one reason modularity plays an important role in open problems involving Dirichlet series, such as the Riemann hypothesis. Indeed, the L-function of the Jacobi theta function, a weight $1/2$ modular form, is the Riemann zeta function: $\theta(\tau) := \sum_{n \in \Z} e^{\pi i n^2 \tau} \leftrightarrow L_{\theta}(s) = 2\zeta(2s)$.

The L-function, $L_f(s)$, is related to the {\em L-integral}, $L_f^*(s)$, which we define to be:
\begin{align}
L_f^*(s):= \int_0^{\infty} \frac{dt}{t} t^s(f(it) -c_f(0)) = \frac{\G(s)}{(2\pi)^s} L_f(s)~.
\end{align}
If $f$ is a modular form without poles and $f(\tau) = \tau^k f(-1/\tau)$, then $L_f^*(s) = i^k L_f^*(k-s)$~\cite{01-Apostol}. When $f$ has poles, $L_f^*(s)$ must be regularized before it is well-defined. In this paper, we focus on modular forms with poles at finite distances above the real-$\tau$ axis. Unfortunately, our results do not capture the previous generalization of L-functions for modular forms whose poles are exclusively at infinity considered by Briunier et al~\cite{03-BFI} and Bringmann et al~\cite{04-BFK, 05-BFK2}. The paper has two parts: 

The first part of the paper is mathematical in nature. In section~\ref{secLfunction}, we regularize $L_f^*(s)^\reg$ associated to a meromorphic modular form $f$. We prove that it satisfies this functional equation, and use it to define its associated L-function, $L_f^\reg(s):= L_f^*(s)^\reg (2 \pi)^s/\G(s)$. Further, we prove that if $f$ is a meromorphic modular form, $e^{+n \pi \sqrt{3}} \lesssim |c_f(n)|$. In section~\ref{secSV}, we evaluate $L_f^\reg(s)$ at $s = 0$ and $s = 1$, when $f$ is a generating function for traces of singular moduli, $\partial_{\tau} \log {\cal H}_d(j(\tau))$~\cite{02-TSM}, and perform a consistency check on $L_f^\reg(s)$. 

The second part of the paper applies these mathematical results to quantum field theory (QFT). In section~\ref{secQFT}, we use the special values of the L-functions found in sections~\ref{secLfunction} and~\ref{secSV} to show that central charges of many two-dimensional conformal field theories (2d CFTs) be written {\em directly} as a tally of states in the theory. The terms involved in these sums increase exponentially with $n$. Hence, this statement crucially hinges on our L-function regularization. Further, we use these L-functions to explicitly verify two sum-rules motivated by T-reflection~\cite{06T-rex1, 07T-rex2}. (We derive a simpler version of these sum-rules that only uses $\zeta(s)$ in Appendix~\ref{secEg1}.) The stress-tensor dictates both the central-charge sum-rule and the T-reflection sum-rule. This is striking as the T-reflection phase, and its associated sum-rule, seems to be a new gravitational anomaly~\cite{06T-rex1, 07T-rex2}. Connecting it to the stress-tensor, which also tracks the conformal anomaly~\cite{08-Yellow}, seems natural. We conclude in section~\ref{secEnd}.

\subsection{Motivations from field theory to study meromorphic modular forms}

Three physics considerations are the driving motivation for this work. First, there is a persistent connection between the behavior of quantum chromodynamics (QCD) at low energies and the behavior of strings. Roughly, the confining potential for QCD can be modeled by a string of finite tension. As this is a model, it does not capture all of the features of QCD, particularly at high energies. Yet, there are some commonalities. 

Heuristically, vacuum loops of closed QCD-strings have the topology of a two-torus. Thus, we may expect the one-loop path integral for low-energy, low-temperature QCD to be modular invariant. (See Refs.~\cite{09-S1S3, 10-2d4d1, 11-2d4d2} for an explicit example.) However, low-energy QCD is known to have an exponential Hagedorn growth in the spectrum: $d(E) \sim E^{p-1} e^{+\beta_H E}$~\cite{12-Hagedorn1, 13-Hagedorn2}. This leads to poles in the path integral $Z(\beta)$ when the inverse temperature, $\beta$ nears $\beta_H$: 
\begin{align}
Z(\beta) = \sum_E d(E) e^{-\beta E} \sim \sum_E E^{p-1} e^{(\beta_H-\beta)E}\to \frac{1}{(\beta-\beta_H)^p}~.
\end{align}
Juxtaposing the expected modularity of the low-energy QCD path integral with the presence of Hagedorn poles suggests that the theory of meromorphic modular forms may give an interesting set of tools to study low-energy QCD. If $d(E)$ exhibits Hagedorn growth, then current ``zeta-function'' technology cannot regularize the Casimir energy sum, $\tfrac12\sum_E d(E)~E$. An L-function for a meromorphic modular form, however, might.

Second, on a related note, L-functions for meromorphic modular forms may allow us to rewrite central charges for unitary CFTs in terms of a direct sum over states in the spectrum, even if the spectrum exhibits Hagedorn growth. While central charges are a measure of the degrees of freedom in the theory, which decrease monotonically along RG flows that progressively integrate-out modes on short distance scales, it is pleasing to see a central charge written directly as a tally of states. Our L-functions allow us to do this.

Third, we would like to test a prediction of the conjectured invariance of QFT path integrals under reflecting their temperatures to negative values (T-reflection). Based on physical arguments, Refs.~\cite{14T-rex0, 06T-rex1}, conjectured that if a path integral $Z(\tau)$ is a weight-$k$ modular form that can be written as an infinite product of the form,
\begin{align}
Z(\tau) = q^{-\Delta} \prod_{n = 1}^{\infty} (1-q^n)^{-d(n)}~,
\end{align}
where $q := e^{2 \pi i \tau}$ and the $d(n)$ are integers for all $n$, then we should have
\begin{align}
\Delta = -\frac{1}{2} \reg \sum_{n = 1}^{\infty} n^1~d(n)~~,~~
~k = \reg\sum_{n = 1}^{\infty} n^0~d(n)~~.
\end{align}
Now, due to the interpretation of the path integral in terms of trace over spectra in the QFT, $\Delta$ has the interpretation as the energy of the lowest-energy state in the QFT. However, save when $Z(\tau)$ is a quotient of Jacobi theta functions or Dedekind eta functions, these $d(n)$-exponents exhibit Hagedorn growth, and the regularization procedure ``$\reg$'' is not known. Our L-functions for meromorphic modular forms gives us a precise context in which we may regulate these sums. Using them, we confirm these predictions from~\cite{14T-rex0, 06T-rex1}.

\subsection{Summary of main mathematical and physical results}

We now summarize our main results. First, in section~\ref{secLfunction}, we prove the connection between poles and Hagedorn growth of $q$-series coefficients, and then provide a well-defined regularized L-function that can be attached to every meromorphic modular form. The following fact establishes Hagedorn growth for meromorphic modular forms:

\begin{prop}
\label{Prop1}
Let $f$ be a meromorphic modular form with a pole of order $p$ at some finite location (away from cusps). There exist real numbers $t \geq \sqrt{3}/2$ and $C > 0$ such that $f(\tau) = \sum_n c_f(n) q^n$ has $C n^{p-1} e^{+2 \pi n t} \leq |c_f(n)| \leq e^{2 \pi n (t+\e)}$ for infinitely many positive integers $n$, for any $\e > 0$.
\end{prop}

This qualitative feature distinguishes meromorphic modular forms from their weakly holomorphic cousins and has played prominently in many papers, such as the last joint paper between Hardy and Ramanujan~\cite{15-HR1919}, and in the recent work of Bialek, Berndt and Ye~\cite{16-BialekPol1, 17-BialekPol2, 18-BialekPol3} and Bringmann, Kane, Lobrich, Ono, and Rolen~\cite{19-KolnPol1, 20-KolnPol2, 21-DivisorsMF1}. However, this result may have escaped broader notice. It is a very important connection between meromorphic modular forms and field theory limits of string path integrals, which often feature Hagedorn growth and Hagedorn poles. We have highlighted it for this reason.

The following statement gives the (polar part) of a well-defined regularization of the integral $L_f^*(s)^\reg$ when $f$ is a meromorphic modular form:

\begin{thm}\label{Thm1}
Let $f$ be a weight $k$ meromorphic modular form that is regular at cusps. Then it's L-integral, $L_f^*(s)^\reg$, evaluates to
\begin{align} 
L_f^*(s)^\reg=-c_f(0) \bigg(\frac{B^s}{s} +\frac{i^k B^{k-s}}{k-s} \bigg) + {\rm regular}(s) = i^k L_f^*(k-s)^\reg~,
\label{eqThm1}
\end{align}
where $0< B < \infty$ and ``${\rm regular}(s)$'' denotes terms that are regular for finite $s \in \C$. 
\end{thm}
\noindent Theorem~\ref{Thm1} is a fusion of Theorem~\ref{ThmLfn} and Corollary~\ref{CorLfn1}, and is one of the main result of this paper. In sections~\ref{secSV} and~\ref{secQFT}, we use $L_f^*(s)^\reg$ to establish the following two results:

\begin{thm}\label{Thm2}
Suppose that $F(\tau)$ is a meromorphic modular form of weight $k$. The L-function for $\langle T(\tau) \rangle := q \partial_q \log F(\tau)$ is,
\begin{align}
L_{\langle T \rangle}^\reg(s) = -\frac{(2 \pi)^s}{\G(s)} \bigg[ 
\Delta\bigg( \frac{B^s}{s}\bigg) - 
\frac{k}{2 \pi} \bigg(\frac{B^{s-1}}{s-1}\bigg) + 
\Delta\bigg(\frac{B^{2-s}}{s-2}\bigg) 
+ {\rm regular}(s) \bigg]~,
\label{eqThm2}
\end{align}
where $0 < B < \infty$ and again ``${\rm regular}(s)$'' denotes terms that are regular for finite $s \in \C$.
\end{thm}

\begin{cor}\label{Cor1}
Suppose that $F(\tau)$ is a meromorphic modular form of weight $k$, and can be written as $q^{-\Delta} \prod_n (1-q^n)^{-d(n)}$. Define $f(\tau):=\sum_n d(n) q^n$. Then,
\begin{align}
\begin{split}
\lim_{s \to -1} L_{f}^\reg(s) 
&= \lim_{s \to -1} \reg \sum_{n = 1}^{\infty} d(n)~n^{-s} 
:= -\lim_{s \to -1} \frac{L_{\langle T \rangle}^{\reg}(s+1)}{\zeta(s+1)}
= +2 \Delta~,  \\
\lim_{~s \to 0~} L_{f}^\reg(s) 
&= \lim_{~s \to 0~} \reg \sum_{n = 1}^{\infty} d(n)~n^{-s} 
:= -\lim_{s \to 0~} \frac{L_{\langle T \rangle}^{\reg}(s+1)}{\zeta(s+1)}
= k ~. \label{eqCor1}
\end{split}
\end{align}
\end{cor}

In section~\ref{secQFT}, we argue that if a 2d CFT path integral equals $F(\tau)$, then it can be interpreted as the partition function for the CFT in the grand canonical ensemble, where particle number is not fixed. Further, we argue that when the CFT is free, then $f(\tau)$ can be interpreted as the partition function for a single-particle excitation in the CFT, i.e. the partition function for the CFT in the canonical ensemble. When a 2d CFT has a path integral with modular weight $k$ and an infinite product expansion, then T-reflection suggests the two sum-rules in Eq.~\eqref{eqCor1}~\cite{06T-rex1, 07T-rex2}. We use the L-functions in Eq.~\eqref{eqThm1} to verify this ``prediction'' from T-reflection. This is one of the main physics-results in this paper.

The other main physics-results that come from our L-functions derive from the interpretation of the quantities $\Delta$ and $k$ in Eqs.~\eqref{eqThm2} and~\eqref{eqCor1}. Namely, in the context of a unitary 2d CFT, $\Delta$ is the Casimir energy and is directly proportional to the central charge of the theory, $c$: $\Delta = c/24$. In writing $\Delta$ as the special value of the L-function of what we call the single-particle partition function for the CFT, which in section~\ref{secQFT} we denote $Z_{\rm CAN}(\tau)$, we directly equate the Casimir energy and thus the central charge to a direct tally of the number of states in the CFT. Because the number of states in the canonical ensemble has an exponential Hagedorn growth, the only way that we make this statement is by using the L-functions for meromorphic modular forms. In this precise sense, this is a new statement. 

%Further, we show that the two formal sum-rules in Eq.~\eqref{eqCor1} are related to 

The main mathematical application of our L-functions in this paper is in the context of traces of singular moduli. All of the notation in this paper on traces of singular moduli follows Zagier's paper of the same name~\cite{02-TSM}. To state the application, let $H(d)$ denote the Hurwitz class number of a quadratic of negative discriminant $d$, $t_n(d)$ be the trace of the unique modular function $J_m(\tau) = q^{-m} + {\cal O}(q) \in M_0^!$ on the {\em CM points discriminant-$d$} called $\A_Q$ (here $M_k^!(\G)$ is the space of weight $k$ modular forms that diverge at cusps of $\PSL_2(\Z)$), which are defined to be the solutions to discriminant $-d$ quadratics that lie within the fundamental domain ${\cal F} := \{ z \in \C \mid |z| \geq 1, {\rm Im}(z) > 0, |{\rm Re}(z)| \leq 1/2\} $, $w_Q$ be the order of the stabilizer of $\alpha_Q$, and let $f_d$ be the unique element in the Kohnen plus-space of $M_{1/2}^!(\G_0(4))$ with Fourier coefficient development  $f_d(\tau)= q^{-d} + {\cal O}(q)$. In Theorems 3 and 5 of~\cite{02-TSM}, Zagier showed that the $d^{th}$ Hilbert polynomial ${\cal H}_d(j(\tau))$ (which is a weakly holomorphic modular function that) satisfies
\begin{align}
{\cal H}_d(j(\tau)) 
:&= ~\prod_{Q \in {\cal Q}/\G} (j(\tau) - j(\alpha_Q))^{1/w_Q} \label{eqHd2}
= q^{-H(d)} ~ {\rm Exp}\bigg[ - \sum_{n = 1}^{\infty} t_n(d) \frac{q^n}{n} \bigg]  \\
&= ~q^{-H(d)} \prod_{n = 1}^{\infty} (1-q^n)^{A(n^2,d)} \in M_0^!
\end{align}
where $j(\tau) := J_1(\tau) + 744$ and $A(n^2,d)$ is the coefficient of $q^{n^2}$ in the Fourier decomposition of,
\begin{align}
f_d(\tau) &= \frac{1}{q^d} + \sum_{n = 1}^{\infty} A(n,d) q^n \in M_{1/2}^!(\G_0(4)).
\end{align}
Let $F(\tau) = {\cal H}_d(j(\tau))$ in Theorem~\ref{Thm2} and Corollary~\ref{Cor1} and call $\langle T (\tau) \rangle = \Lambda_d(\tau)$. Then,
\begin{align}
\lim_{s \to 0} L_{\Lambda_d}^\reg(s) = \lim_{s \to 0} \reg \sum_{n = 1}^{\infty} \bigg( \sum_{m|n} m A(m^2,d) \bigg) n^{-s} = \lim_{s \to 0} \reg \sum_{n = 1}^{\infty} \frac{t_n(d)}{n^s} = -H(d)~.
\end{align}
This relationship between Hurwitz class numbers and traces of singular moduli, while formal, is new and crucially hinges on the definition of L-functions for meromorphic modular forms.

Finally, we point-out that the exponential growth in Proposition~\ref{Prop1} gives an amusing way to effectively determine $H(d)$ without direct reference to quadratics. Explicitly:

\begin{thm}\label{Thm3}
Consider the two expressions for $\Lambda_d(\tau) := q \d_q {\cal H}_d(j(\tau)) \in F_2$ from the two different representations of ${\cal H}_d(j(\tau))$ in Eq.~\eqref{eqHd2}, and let ${\rm Li}_0(x):= \sum_n x^n$. The $q$-series coefficients are $t_n(d)$. The $q$-series coefficients $\tilde{t}_n(d)$ of the pole-subtracted function,
\begin{align}
\tilde{\Lambda}_d(\tau) = \Lambda_d(\tau) - \sum_{Q \in {\cal Q}_d/\G} \frac{1}{w_Q}{\rm Li}_0(e^{2 \pi i (\tau-\alpha_Q)}) = \sum_{n = 0}^{\infty} \tilde{t}_n(d) q^n~, \label{eqThm3}
\end{align}
have exponential growth that is bounded by $|\tilde{t}_n(d)| < e^{\pi n\sqrt{2}}$. 
\end{thm}
We can find $H(d)$ by determining the number of terms that need to be subtracted-out from $\La_d(\tau)$ before its $q$-series coefficients are bounded by $|\tilde{t}_n(d)| < e^{\pi n \sqrt{2}}$.

\section{L-functions for meromorphic modular forms}\label{secLfunction}

Write ${\rm Re}(z)$ and ${\rm Im}(z)$ for the real and imaginary parts of a complex number $z\in \C$. Let $\HH:=\{\tau\in\C\mid {\rm Im}(\tau)>0\}$ denote the upper half-plane. A {\em holomorphic modular form} of weight $k$ (for the full modular group, $\SL_2(\Z)$) is a holomorphic function $f:\HH\to \C$ such that $f(\tau+1)=f(-\frac1\tau)\tau^{-k}=f(\tau)$ for all $\tau\in \HH$, and $f(\tau)={\cal~O}(1)$ as ${\rm Im}(\tau)\to \infty$. Write $M_k$ for the vector space of modular forms of weight $k$. A product of modular forms of weights $k_1$ and $k_2$ is a modular form of weight $k_1+k_2$, so $M_*:=\bigoplus_k  M_k$ is naturally a (graded) ring. 

In this work a {\em meromorphic modular form} of weight $k$ is a quotient $f=\frac{f_1}{f_2}$ where $f_i\in M_{k_i}$, the denominator $f_2$ is not identically zero, and $k=k_1-k_2$. We write $F_k$ for the vector space of meromorphic modular forms of weight $k$, and set $F_*:=\bigoplus_k F_k$. Then $F_k$ vanishes unless $k$ is an even integer, and $F_*$ may be regarded as a ``graded ring of fractions'' of $M_*$. 

Define $F(\infty)$ to be the vector space of meromorphic functions on $\HH$ that satisfy $f(\tau)=f(\tau+1)$, have only finitely many poles in any compact subset of the vertical strip 
\begin{align}
V:= \{\tau\in \HH\mid {\rm Im}(\tau) > 0~,~|{\rm Re}(\tau)| \leq \tfrac12\} \subset \HH,
\end{align} 
and satisfy the growth condition $f(\tau)={\cal O}(e^{C{\rm Im}(\tau)})$ as ${\rm Im}(\tau)\to \infty$ for some $C>0$. Then $F_*$ is a subspace of $F(\infty)$. Set $e(x):=e^{2\pi i x}$ and $q:=e(\tau)$. For any $f\in F(\infty)$ we write $c_f(n)$ for the coefficient of $q^n$ in the Fourier expansion 
\begin{align}\label{eqCfn}
	f(\tau)=\sum_{n\in\Z}c_f(n)q^n.
\end{align}
By the growth condition on $F(\infty)$ we have $c_f(n)=0$ for $n\ll 0$. 

If $f\in M_k$ satisfies $f(\tau)={\cal O}(q)$ (i.e. if $f$ is a holomorphic cusp form) then $c_f(n)={\cal O}(n^k)$ as $n\to \infty$, and the Dirichlet series 
\begin{align}
L_f(s):=\sum_{n\in\Z}\frac{c_f(n)}{n^s} \label{eqLf}
\end{align}
converges absolutely for ${\rm Re}(s)>k+1$~\cite{01-Apostol}. But as we show in Lemma \ref{LemRegC3}, general $f\in F_k$ have $c_f(n)$ that grow exponentially and the right-hand side of \eqref{eqLf} is nowhere convergent. 

In this section we introduce a regularization of \eqref{eqLf}, $L_f^\reg(s)$, that is well-defined and analytic in $s$ for an arbitrary meromorphic modular form $f\in F_k$. To introduce it, recall that when $f \in M_k$ , the L-integral is given by:
\begin{align}
L_f^*(s) = \int_{0}^{\infty} \frac{dt}{t} t^s ~\big(f(it) - c_f(0)\big)~.
\label{eqLint0}
\end{align}
Here, when $f \in M_k$, the integral and the $q$-series commute and yield the sum representation of $\G(s) L_f(s)/(2\pi)^s$ from Eq.~\eqref{eqLf}. The analogous statement for $f \in F_k$ requires regularization.

Generic $f \in F_k$ have poles within the fundamental domain ${\cal F}:=\{ \tau \in \C \mid 0 < {\rm Im}(\tau), |\tau| \geq 1, |{\rm Re}(\tau)| \leq 1/2 \}$ and poles at $i \infty$. For this case, we define the L-function, $L_f(s)$, by a deformation of the associated L-integral. Let $L_f^*(s)^{\reg}$ be the following limit of the regularized contour integral:
\begin{align}
L_f^\reg(s) := \frac{(2\pi)^s}{\G(s)} L_f^*(s)^{\reg} 
:= \frac{(2\pi)^s}{\G(s)}  
\lim_{\e \to 0} 
\lim_{\La \to \infty} 
\frac{1}{2}\sum_{\pm \e} 
\reg \int_{\g(\La,\pm\e)}
\frac{d\tau}{\tau} 
\left( \frac{\tau}{i}\right)^s 
~\big(f(\tau)-c_f(0)\big)~,
\label{eqLint}
\!\!\!
\end{align}
where $\g(\La,\pm\e)$ begins at $\tau = i/\Lambda$, moves upwards just to the left (right) of the imaginary-$\tau$ axis, and ends at $\tau = i\Lambda$. (As the integrand is meromorphic in $\tau$, then if $\g(\La,\e)$ passes between the imaginary-$\tau$ axis and any pole off of the axis, and if the endpoints are fixed, then the precise path of $\g(\La,\e)$ is not needed.) Deformations away from the ${\rm Im}(\tau)$-axis allow us to deal with poles along the imaginary-$\tau$ axis. Crucially, if $f \in M_k$ then Eqs.~\eqref{eqLf},~\eqref{eqLint0} and~\eqref{eqLint} match.

The remainder of this section is as follows. In section~\ref{secContourL} we describe this contour and make contact with a previous definition of $L^{\rm reg}_f(s)$ for the special case $f \in M_k^!$. In section~\ref{secGrowLF} we prove that if $f \in F_k$ has a single order-$p$ pole at $\tau = s + i t \in {\cal F}$, then $|c_f(n)| \sim C n^{p-1} e^{2 \pi n t}$ for some $C$. In section~\ref{secSeriesL} we rewrite the integral transform in Eq.~\eqref{eqLint} explicitly in terms of the $q$-series coefficients of $f \in F_k$. In section~\ref{sec15-HR1919}, we comment on the numerical convergence of $L_f^\reg(s)$ and on path dependence.

\subsection{Defining the contour and defining the regularization}\label{secContourL}

We define the regularized L-integral, $L_f^*(s)^{\reg}$, by the contour $\g(\La,\e)$, which begins at $i/\La$ and goes to $i \La$, while keeping just $\e$ to the right of poles on the imaginary-$\tau$ axis. 

\begin{lem}\label{LemReg1}
Let $f \in M_k$ and $B$ be a positive real number. Then the L-integral, $L_f^*(s)^{\reg}$ is
\begin{align}
\!\!\!\!L^*_f(s)^\reg:=&\lim_{\Lambda \to \infty} \int_{\g\left(\Lambda,0\right)} \frac{d\tau}{\tau} 
\left( \frac{\tau}{i}\right)^s 
~\big(f(\tau)-c_f(0)\big) 
= \frac{\G(s)}{(2\pi)^s} \sum_{n =1}^{\infty} \frac{1}{n^s} \label{eqLemReg1a} \\
=&-c_f(0)\left( \frac{B^s}{s} + \frac{i^k~B^{s-k}}{k-s} \right) + \sum_{n \neq 0} c_f(n) \left( \frac{\G(s,2\pi n B)}{(2\pi n)^s} + i^k \frac{\G(k-s,2\pi n/B)}{(2\pi n)^{k-s}} \right),\!\!\!\!
\label{eqLemReg1b}
\end{align}
where $\G(s,x) = \int_x^{\infty} dt~ t^{s-1} e^{-t} $. This matches $L_f(s)$ defined in Eq.~\eqref{eqLf}. As $c_f(n) = {\cal O}(n^{2k-1})$ for $f \in M_k$~\cite{01-Apostol} and as $\G(s,x) \sim x^{s-1} e^{-x}$ for large-$x$, the sum in Eq.~\eqref{eqLemReg1b} converges absolutely for any finite $s \in \C$ for every $f \in M_k$. 
\end{lem}

\begin{proof}
As $f \in M_k$ does not have any poles, then any integration contour from $i/\Lambda$ to $i\Lambda$ yields the same result. So the exact shape of our particular contour, $\g(\La,0)$, from $i/\La$ to $i \La$ does not have any effect on the final result. As $f \in M_k$ is bounded along this finite contour, the contour integral converges and we may write 
\begin{align}
\int_{t = 1/\Lambda}^{t = \Lambda} \frac{dt}{t} t^s \big( f(it) -c_f(0) \big) = 
\sum_{n =1}^{\infty} c_f(n) \int_{1/\Lambda}^{\Lambda} \frac{dt}{t} t^s e^{-2 \pi n t} ~.
\end{align}
As $\Lambda \to \infty$, the integrals in the sum on the RHS evaluate to $\G(s)/(2\pi n)^{s}$. Thus the sum over $c_f(n) \times$ integrals yields Eq.~\eqref{eqLemReg1a}, which matches the definition of $L_f(s)$ in Eq.~\eqref{eqLf}. 

Now, we can split-up the integral from $t = 1/\La$ to $t = \La$ as the sum of an integral from $t = 1/\La$ to $t = B$ and an integral from $t = B$ to $t = \La$.  Noting $f(-1/\tau) = \tau^k f(\tau)$, we can rewrite the integral from $t = 1/\La$ to $t = B$  in the following way:
\begin{align}
\begin{split}
L_f^*(s|\Lambda)^\reg =
&+\int_{B}^{\Lambda} \frac{dt}{t} (t^s) \big( f(it) -c_f(0) \big)+i^k \int_{1/B}^{\Lambda} \frac{dt}{t}  (t^{k-s}) \big( f(it) -c_f(0) \big) \\
&-c_f(0) \int_{\tfrac1\Lambda}^{B} \frac{dt}{t} \big( t^{+s} - i^k t^{s-k} \big)
\end{split}
\label{eqT0Lf}
\end{align}
Note that this process is independent of the value of $B$. Freedom to choose $1/\La < B < \La$ plays an important role in section~\ref{secSeriesL}, where we derive explicit forms for $L_f^*(s)^\reg$ for $f \in F_k$. 

Now, for any fixed $\Lambda \gg 1$ we may exchange the integration and the $q$-series summation, in Eq.~\eqref{eqT0Lf}. Recalling the gamma-function $\G(s,x):= \int_x^{\infty} dt~t^{s-1} e^{-t}$, we find 
\begin{align}
L_f^*(s|\Lambda)^\reg = &
-c_f(0) \left( i^k  \frac{B^{s-k}-\La^{s-k}}{k-s}\right) + \sum_{n \neq 0} c_f(n) \left(\frac{\G(s,2\pi n B)-\G(s, 2 n \pi \La)}{(2\pi n)^s} \right)
\label{eqLemReg1e1}
 \\
&
-c_f(0) \left( \frac{B^{s}-\La^{s}}{s}\right) + \sum_{n \neq 0} c_f(n) \left(i^k \frac{\G(k-s,2\pi n/B)-\G(k-s, 2 n \pi \La)}{(2\pi n)^{k-s}} \right)
\nonumber
\end{align}
When ${\rm Re}(s) > 0$ and ${\rm Re}(s-k)> 0$, we find that the $\La \to \infty$ limit of Eq.~\eqref{eqT0Lf} equals
\begin{align}
\lim_{\Lambda \to \infty} L_f^*(s|\Lambda)^\reg =  -c_f(0)\left( \frac{B^s}{s} + \frac{i^k~B^{s-k}}{k-s} \right) + \sum_{n \neq 0} c_f(n) \left( \frac{\G(s,2\pi n B)}{(2\pi n)^s} + i^k \frac{\G(k-s,2\pi n/B)}{(2\pi n)^{k-s}} \right).
\end{align}
We finish the proof by analytically continuing $1/s$ and $1/(k-s)$ to the entire $s$-plane, and then noting that the scaling of $\G(s,x) \sim x^{s-1} e^{-x}$ implies the sum converges absolutely for any finite value of $s \in \C$ when $f \in M_k$ and $c_f(n) = {\cal O}(n^{2k-1})$.
\end{proof}

The sum $\sum_n c_f(n) \G(s,2\pi n)/(2 \pi n)^s$ converges absolutely for every finite $s \in \C$ for {\em any} sequence of $c_f(n)$ whose exponential growth is $|c_f(n)| \lesssim {\rm max}\{e^{2 \pi  n/B},e^{2 \pi  n B}\}$. Now, we note:

\begin{lem}\label{LemReg2}
Let $f \in M_k^!$, and $\La \gg 1$, and consider the regularized contour integral $L_f^*(s|\La)^\reg:=\int_{\g(\La,\e)} d(\tau/i) (\tau/i)^{s-1} (f(\tau) -c_f(0))$. For any fixed $\La$, it evaluates to:
\begin{align}
\!\!\!\!\!\!\!\!
L_f^*(s|\La)^\reg 
=  -c_f(0)\left( \frac{B^s}{s} + \frac{i^k~B^{s-k}}{k-s} \right) 
&+ \sum_{n \neq 0} c_f(n) \left( \frac{\G(s,2\pi n B)}{(2\pi n)^s} + i^k \frac{\G(k-s,2\pi n/B)}{(2\pi n)^{k-s}} \right)\!\!\!\!
\label{eqLemReg2}\\
&- \sum_{n \neq 0} c_f(n) \left( \frac{\G(s,2\pi n \La)}{(2\pi n)^s} + i^k \frac{\G(k-s,2\pi n \La)}{(2\pi n)^{k-s}} \right).\!\!\!\!
\label{eqHouston}
\end{align}
As there exists $C>0$ for any $f\in M_k^!$ where $|c_f(n)| \ll e^{C \sqrt{n}}$ for $n \gg 1$, this sum in Eq.~\eqref{eqLemReg2} converges absolutely for finite $s \in \C$. The explicit result in Eq.~\eqref{eqLemReg2} matches $L_f^*(s)^{\reg}$ in Refs.~\cite{03-BFI, 04-BFK, 05-BFK2}. Yet, each of the finite set of terms in Eq.~\eqref{eqHouston} with $n < 0$ diverge as $\La\to \infty$.
\end{lem}

\begin{proof}
In this proof, we reproduce Eq.~\eqref{eqLemReg2} and make contact with the previously defined L-functions for weakly holomorphic modular forms in~\cite{03-BFI, 04-BFK, 05-BFK2}. 

First, we show that our proof of Eq.~\eqref{eqLemReg1b} in Lemma~\ref{LemReg1} applies without essential modification for any finite $\La \gg 1$: As $f \in M_k^!$ only has poles at cusps, all integration contours from $\tau = i/\Lambda$ to $\tau = i\Lambda$ yield the same result. So again we can path-deform either of the two finite-$\Lambda$ integrals that appear in Eq.~\eqref{eqLemReg2} to the much simpler integral $\int_{i/\Lambda}^{i \Lambda} dt ~t^{s-1} (f(it)-c_f(0))$. 

We now exploit the modularity to again rewrite the nontrivial $\La$-regulated integral as in Eq.~\eqref{eqLemReg1e1}. When ${\rm Re}(s) > 0$ and ${\rm Re}(s-k) > 0$, we again find that the integral is: 
\begin{align}
\label{eqLfReg}
L_f^*(s|\Lambda)^\reg =
&
-c_f(0) \left( \frac{B^{s}-\La^{s}}{s}\right) 
+ \sum_{n \neq 0} c_f(n) \left(\frac{\G(s,2\pi n B)-\G(s, 2 n \pi \La)}{(2\pi n)^s} \right)
\\
&
-c_f(0) \left( i^k  \frac{B^{s-k}-\La^{s-k}}{k-s}\right) 
+ \sum_{n \neq 0} c_f(n) \left(i^k \frac{\G(k-s,2\pi n/B)-\G(k-s, 2 n \pi \La)}{(2\pi n)^{k-s}} \right) \nonumber
\end{align}
Now, note that the $\La$-independent terms in Eq.~\eqref{eqLfReg} exactly match those in Eq.~\eqref{eqLemReg1a}, and exactly match those in~\cite{04-BFK, 05-BFK2}. However, if $f$ has a $q$-series that begins with $q^{-N}$ for some positive integer $N$, then there exists a finite number of non-zero coefficients $c_f(n)$ for $n < 0$. In the sum over $n \neq 0$, these terms multiply incomplete gamma functions $\G(\s,2 \pi n \La)$, which diverge exponentially as $\La \to \infty$. These terms ruin the convergence of $L_f^*(s|\La)^\reg$ as $\La \to \infty$.

If not for these terms, then the terms in Eq.~\eqref{eqLemReg2} would dominate $L_f^*(s|\La)^\reg$ in the $\La \to \infty$ limit, and we would recover the results of~\cite{04-BFK, 05-BFK2}. This finishes the proof.
\end{proof}

Crucially, this does not reproduce the L-functions for $f \in M_k^!$ in~\cite{03-BFI, 04-BFK, 05-BFK2}. If we could equate the regularized L-function with the constant term of $L_f^*(s|\La)^\reg$ in a large-$\La$ expansion, then we would reproduce their results. However, we are unaware of any reason to justify such an equality. Hence, we restrict our attention to $f \in F_k$ that are regular at cusps. Within this space, we prove the following Lemma: 
%We did not use the regulator of~\cite{03-BFI, 04-BFK, 05-BFK2} as is not holomorphic. This makes it unclear how to use residue theorems to isolate the divergences associated with poles at finite locations. 
%We now allow $f$ to be a meromorphic modular form, $f \in F_k$. 

\begin{lem}\label{LemReg4}
Let $f \in F_k$ be regular at cusps. Then the $\La \to \infty$ and $\e \to 0$ limits of
\begin{align}
L_f^*(s|\Lambda,\e)^\reg := \sum_{\pm \e} \frac{1}{2}\int_{\g\left(\Lambda,\pm\e\right)} \frac{d\tau}{\tau} 
\left( \frac{\tau}{i}\right)^s 
\big(f(\tau)-c_f(0)\big)~.\label{eqLemReg4}
\end{align}
exist and yield a finite result. Hence $L_f^\reg(s)$ as defined in Eq.~\eqref{eqLint} exists for generic $f \in F_k$.  
\end{lem}

\begin{proof}
We require the outermost contours $\g(\La,\pm\e)$ to pass to left and to the right of each on-axis pole exactly once, while enclosing no poles off the axis. %The existence of such an $\e > 0$ is guaranteed by the fact that every $f \in F_k$ will have poles at a finite number of points within ${\cal F}$. 
Now for every $f \in F_k$ there is a finite integer $N$ such that
\begin{align}
f(\tau) = \sum_{n = N}^{\infty} c_f(n) e^{2 \pi i n \tau} \in F_k~.
\end{align}
If $N \geq 0$, then $|f(it) -c_f(0)| = {\cal O}(e^{-2 \pi t})$ as $t \to \infty$. Hence, $|f(\tau)|$ is bounded along $\g(\La,\e)$ and exponentially decays as $\tau \to i \infty$ (and as $\tau \to i 0$). In the limit $\e \to 0$, each contour encloses poles on the axis but remains finite as $\La \to \infty$.  Because the two contour integrals differ by a finite sum of residues, when $f$ is regular at $i \infty$, then $\lim_{\e \to 0} (\sum_{\pm \e} \tfrac12 \lim_{\La \to \infty} L_f^*(s|\La,\e))$ exists and yields a finite result. Thus $L_f^*(s)^\reg$ exists and is finite when $f \in F_k$ is regular at cusps. 
\end{proof}

We now comment that if the regularization prescription in Lemma~\ref{LemReg2} is applied to modular forms with poles at cusps, $f\in M_k^!$, then proof of Lemma~\ref{LemReg4} would generalize to all $f \in F_k$. The reason is as follows. 

Let $f \in F_k$ have a $q$-series that begins with $q^N (1 + {\cal O}(q))$ with $N < 0$. Now, consider any $g \in M_k^!$ whose polar terms match those of $f \in F_k$. More explicitly, $f(\tau) + g(\tau) = {\cal O}(q^0)$. It is relatively simple to construct an explicit function whose polar terms $c_g(n) q^{-n}$ match those of any $f \in F_k$. Let $h \in M_k$ be holomorphic, and let $J_M \in M_0^!$ be the unique modular function that begins with $q^{-M} + {\cal O}(q)$. Then one can choose $|N|$ coefficients $C(k)$ for $k \in \{ 1, \ldots, |N|\}$ such that $g(\tau):= h(\tau) \sum_{k = 1}^{|N|} C(k) J_k(\tau) \in M_k^!$ whose $q$-series begins with $g(\tau) = \sum_{n = N}^{-1} (-c_f(n))q^n + {\cal O}(q^0)$. Thus,
\begin{align}
g(\tau) = \sum_{n = N}^{\infty} c_g(n) e^{2 \pi i \tau} \in M_k^!~,~~{\rm where}~~
f(\tau) + g(\tau) = \sum_{n = 0}^{\infty} (c_f(n) + c_g(n)) e^{2 \pi i \tau} \in F_k~.
\end{align}
Now, $f+g \in F_k$ has poles only within the interior of ${\cal F}$, while $g \in M_k^!$ has poles only at the cusp. Exploiting the linearity of the integral functional, we have $L_f^*(s)^\reg = L_{f+g}^*(s|\La)^\reg - L_g^*(s|\La)^\reg$. If the regularization prescription in Lemma~\ref{LemReg2} worked to regularize the L-integral for $f \in M_k^!$, then we would have $L_f^*(s|\La,\e)^\reg$ equal to a sum of integrals that are regular and convergent as $\La \to \infty$ and $\e \to 0$. This would extend Lemma~\ref{LemReg4} to all $f \in F_k$. 

Yet, when $f \in F_k$ has poles at cusps, there are a finite number of terms in $L_f^*(s|\La,\e)^\reg$ from the leading polar powers of $q$ in the $q$-series that diverge at the upper end of integration as $\G(s,-2\pi n \La)$ as $\La \to \infty$. Thus, as it stands, the regularization only applies to $f \in F_k$ that are regular at cusps. We briefly comment on several aspects of our regularization procedure:
\begin{enumerate}
\item Another contour $\g$ between $i0^+$ and $i \infty$ would define another $L_f^*(s)^\reg_{\g}$. If $f$ has a finite number of poles ``between'' $\g$ and the ${\rm Im}(\tau)$-axis, then $L_f^*(s)^\reg_{\g} - L_f^*(s)^\reg$ differ by a finite sum of residues. Residues of poles to left and right of the ${\rm Im}(\tau)$-axis have opposite signs. See sections~\ref{sec15-HR1919} and~\ref{secEnd} for more discussion of path dependence.
\item The contour-regularization of the pole at the cusp in our proof of Lemma~\ref{LemReg2} is distinguished from the regularization in~\cite{03-BFI, 04-BFK, 05-BFK2}. Here, they deformed $(f(\tau)-c_f(0)) \to (f(\tau)-c_f(0)) e^{2 \pi \up {\rm Im}(\tau)}$. When $f \in M_k^!$, the $\up$-regularized L-functions can be analytically continued to $\up = 0$. We regularize our L-functions by contour-deformation so that we could more easily discuss the case where $f \in F_k$ may have poles {\em within} ${\cal F}$. 
\item Unfortunately, the regularization prescriptions here and in~\cite{04-BFK} are in tension with each other: once the integrand is non-holomorphic then the Cauchy residue theorem no longer applies. Non-holomorphic deformations are central to~\cite{04-BFK}, where they define $L_f^*(s)^\reg$ for $f \in M_k^!$ with poles at cusps. Residue theorems are central for our approach, where we define $L_f^*(s)^\reg$ for $f \in F_k$ with poles away from cusps.
\item The conventional definition, $L_f^*(s) = \int_{0}^{\infty} dt~t^{s-1}~(f(it)-c_f(0))$, can be extended to nontrivial integrands in many ways. For instance, we could define $L_f^*(s)_\g^\reg$ by the non-holomorphic integral $\int_\g dt~t^{s-1}~(f(\tau)-c_f(0))$, where $t = {\rm Im}(\tau)$, or we could regularize the poles by a non-holomorphic deformation of the integrand. 
\item Defining $L_f^*(s)^\reg$ as a meromorphic integrand integrated along contours that pass to {\em both} sides of the ${\rm Im}(\tau)$-axis gives an unambiguous definition of the regulated integrals that occur when poles lie along the ${\rm Im}(\tau)$-axis. By using contour-deformations and residue theorems, we obtain uniquely defined, concrete, and well-behaved expressions for $L_f^*(s)^\reg = \reg \int_{0}^{\infty} dt~t^{s-1}(f(it)-c_f(0))$ for $f \in F_k$ that are regular at cusps.
\end{enumerate}
In section~\ref{secSeriesL} we give explicit formulae for  $L_f^*(s)^\reg$ analogous to the sum over incomplete gamma functions and $q$-series coefficients in Lemmas~\ref{LemReg1} and~\ref{LemReg2}, when $f \in F_k$ is regular at cusps. Before this, in section~\ref{secGrowLF} we discuss the lockstep relationship between exponential growth of the $q$-series coefficients of $f \in F_k$ and the locations of its poles within ${\cal F}$.

\subsection{Poles and exponential growth}\label{secGrowLF}

We now show the relationship between the location of the poles of $f \in F_k$ and the exponential growth of its $q$-series coefficients. To start, let $\RF_{\tau_p}:F(\infty)\to F(\infty)$ denote a linear operator that projects onto the subspace of functions that are regular at $\tau_p$. To explicitly define it, write $f(\tau)=\sum_{n\in \Z} r_{f,\tau_p}(n) (\tau-\tau_p)^n$ for the Laurent series expansion of $f$ in a (sufficiently small) deleted neighborhood of $\tau_p$, and set
\begin{align}
\label{eqRstar}
r^*_{f,\tau_p}(m):=\frac{(-2\pi i)^m}{(m-1)!}r_{f,\tau_p}(-m)
\end{align} 
for $m>0$. We define $\RF_{\tau_P}$, and the related $\PF_{\tau_p}: F(\infty) \to F(\infty)$, as
\begin{align}
\label{eqPtau0f}
\begin{split}
\PF_{\tau_p}f(\tau)&:=\sum_{m>0}r^*_{f,\tau_p}(m)\Li_{1-m}\big(e(\tau-\tau_p)\big)~,\\
\RF_{\tau_p}f(\tau)&:=(I-\PF_{\tau_p})f(\tau)~,
\end{split}
\end{align}
where $\Li_s(z):=\sum_{n>0}z^nn^{-s}$ is the polylogarithm function with order $s$. Note that the summation in \eqref{eqPtau0f} is finite since $r_{f,\tau_p}(n)=0$ for $n\ll 0$. If $f$ is regular at $\tau_p$, so that $r_{f,\tau_p}(n)=0$ for $n<0$, then $\RF_{\tau_p}f=f$. When $m$ is a positive integer, as $z\to 0$ we have 
\begin{align}\label{eqLaurent}
{\rm Li}_{1-m}(e(z))
&= \frac{(m-1)!}{(-2 \pi i z)^{m}} + \sum_{\ell = 0}^{\infty} \frac{\zeta(1-m-\ell)}{(-1)^{m}~\ell!} (-2 \pi i z)^{\ell}~.
\end{align}

\begin{lem}\label{LemRegC1}
For $\tau_p\in \HH$ and $f\in F(\infty)$ the function $\RF_{\tau_p}f$ is regular at $\tau_p$.
\end{lem}

\begin{proof}
If $f$ is regular at $\tau_p$ then $\RF_{\tau_p}f=f$ and the claim is true. Otherwise, combining \eqref{eqRstar}, \eqref{eqPtau0f} and \eqref{eqLaurent} we see that
\begin{align}\label{eqPfrtau0}
	(\PF_{\tau_p}f)(\tau)
	=\sum_{m>0}r_{f,\tau_p}(-m)(\tau-\tau_p)^{-m}+{\cal O}(\tau-\tau_p)
\end{align}
as $\tau\to \tau_p$. The estimate \eqref{eqPfrtau0} holds with $f$ in place of $\PF_{\tau_p}f$ if we also replace ${\cal O}(\tau-\tau_p)$ with ${\cal O}(1)$, so $(\RF_{\tau_p}f)(\tau)=f(\tau)-(\PF_{\tau_p}f)(\tau)={\cal O}(1)$ as $\tau\to\tau_p$.
\end{proof}

Observe that if $X$ is any compact subset of $\HH$ then, since any $f\in F(\infty)$ has only finitely many poles $\tau_p \in X$, we obtain a well-defined projection operator $\RF_X:F(\infty)\to F(\infty)$ by setting $\RF_X:=I-\PF_X$ where 
\begin{align}
(\PF_Xf)(\tau):=
\sum_{\tau_p\in X}\sum_{m>0}r^*_{f,\tau_p}(m)\Li_{1-m}\left(e(\tau-\tau_p)\right),
\end{align} 
and the functions in the image of $\RF_X$ are regular in $X$. In what follows we consider the set $V_Y$ defined by the real number $0 < Y < \infty$:
\begin{align}
V_Y&:=\{\tau\in V \mid -\tfrac12\leq {\rm Re}(\tau)<\tfrac12,\;{\rm Im}(\tau)\geq Y\} \subset V 
\label{eqVY}
\end{align}
When $Y$ is a finite and positive real number, $X = V_Y$ is a rectangular subset of $V$, and we define $\RF_Y f$ and $\PF_Y f$ to respectively be $\RF_{V_Y} f$ and $\PF_{V_Y}f$.

Now consider $f\in F(\infty)$, and note that the imaginary parts of the poles of $f$ are bounded. The next lemma explains how the maximum of these imaginary parts bounds the growth of the Fourier coefficients of $f$:

\begin{lem}\label{LemRegC2}
Let $f\in F(\infty)$ and choose $t_0\in \R$ so that no pole of $f$ has imaginary part greater than $t_0$. Then $c_f(n)e^{-2\pi tn}\to 0$ as $n\to \infty$ when $t>t_0$.
\end{lem}

\begin{proof}
By Parseval's identity we have that
\begin{align}
\int_{-\frac12}^{\frac12} {\rm d}\sigma | f(\sigma + i t)|^2 = \sum_{n} \big| c_f(n) e^{- 2 \pi tn}\big|^2 
\end{align}
converges for any $t>t_0$. The claim follows.
\end{proof}

Next we verify a converse to Lemma \ref{LemRegC2}, showing that the coefficients of $f\in F(\infty)$ grow exponentially when $f$ has poles in $\HH$. 

\begin{lem}\label{LemRegC3}
Let $f\in F(\infty)$, let $t_0$ be the maximal imaginary part of a pole of $f$, and let $m_0$ be the maximal order of a pole of $f$ with maximal imaginary part. Then there exists a constant $C>0$ such that $|c_f(n)|>Cn^{m_0-1}e^{2\pi n t_0}$ for infinitely many positive integers $n$.
\end{lem}

\begin{proof}
Let $t_0$ and $m_0$ be as in the statement of the lemma, and let $\tau_0$ be a pole of $f$ with ${\rm Im}(\tau_0)=t_0$ and $r_{f,\tau_0}(-m_0)\neq 0$ and $r_{f,\tau_0}(n)=0$ for $n<-m_0$. Set $g(\tau):=\RF_{\tau_0}f$. Then we have $c_g(n)=c_f(n)-C_{m_0}n^{m_0-1}e^{-2\pi i n\tau_0}$ where $C_{m_0}=\frac{(-2\pi i)^{m_0}}{(m_0-1)!}r_{f,\tau_0}(-m_0)$. Now $g(\tau)$ is regular at $\tau=\tau_0$, so the Fourier series $g(\tau)=\sum c_g(n)e^{2\pi i n\tau}$ converges absolutely at $\tau=\tau_0$, so $|c_g(n)|e^{-2\pi n t_0}\to 0$ as $n\to \infty$. In particular, for any $\e>0$ we have  $\e>|c_g(n)|e^{-2\pi n t_0}$ for $n$ sufficiently large. Take $C=\frac12C_{m_0}$ and $\e=\frac12 C$. 

If $|c_f(n)|\leq Cn^{m_0-1}e^{2\pi n t_0}$ for sufficiently large $n$ then we have
\begin{align}
\begin{split}
	\tfrac12 C >&\,|c_f(n)-2Cn^{m_0-1}e^{-2\pi in \tau_0}|e^{-2\pi n t_0}\\
	\geq&\, 
	(2Cn^{m_0-1}e^{2\pi n t_0}-|c_f(n)|)e^{-2\pi nt_0}\\
	\geq&\,
	Cn^{m_0-1}
\end{split}
\end{align}
for $n$ sufficiently large. This is a contradiction, so we must have $|c_f(n)|>Cn^{m_0-1}e^{2\pi n t_0}$ for infinitely many positive integers $n$.
\end{proof}

Lemmas~\ref{LemRegC1},~\ref{LemRegC2}, and~\ref{LemRegC3} establish a direct correspondence between $f \in F(\infty)$ having poles of order-$m$ at $\tau_p$ and terms in its $q$-series that exponentially as $n^{m-1} (e^{-2\pi i \tau_p})^n$. We will find it useful to associate every function $f \in F(\infty)$ with two positive real numbers, which we denote $U$ and $D$ and define as follows:
\begin{align}
D &:={\rm min}\{ ~{\rm Im}(\tau) \mid 1/f(\tau) = 0 ~,~ {\rm Im}(\tau) \geq \sqrt{3}/2 ~\} \in \R~,
\label{eqDD}
\\
U &:={\rm max}\{~ {\rm Im}(\tau) \mid 1/f(\tau) = 0 ~,~ {\rm Im}(\tau) \geq \sqrt{3}/2 ~\} \in \R~.
\label{eqDU}
\end{align}
The bound of $U,D \geq \sqrt{3}/2$ ultimately comes from the fact that if $\tau \in {\cal F}$, then ${\rm Im}(\tau) \geq \tfrac{\sqrt{3}}{2}$. This bound on $D$, together with Lemmas~\ref{LemRegC1},~\ref{LemRegC2}, and~\ref{LemRegC3}, implies the following result: 

\begin{prop}
\label{PropGrowthCap}
Let $f \in F_k$ have a pole of order $p > 0$ within $\HH$ away from the cusp. Then there exist real numbers $t \geq \tfrac{\sqrt{3}}{2}$ and $C > 0$ such that the $c_f(n)$ in
$f(\tau) = \sum_n c_f(n) q^n$ obey $C n^{p-1} e^{+2 \pi n t} \leq |c_f(n)| \leq e^{2 \pi n (t+\e)}$ for infinitely many positive integers $n$, for every $\e > 0$.
\end{prop}

\begin{proof}
If $f \in F_k$ has an order-$p$ pole at $\tau = s + i t\neq i \infty$ for $t>0$ and $s,t \in \R$, then Lemmas~\ref{LemRegC1}, ~\ref{LemRegC2}, and ~\ref{LemRegC3} imply that a real number $C > 0$ exists, such that the $q$-series coefficients $c_f(n)$ are bounded by $C n^{p-1} e^{+2 \pi n t} \leq |c_f(n)| \leq e^{2 \pi n (t+\e)}$ for every $\e > 0$. As $f$ is modular, it has at least one pole within ${\cal F}$. Hence there is an order-$p$ pole with $t \geq \tfrac{\sqrt{3}}{2}$. 
\end{proof}

Consider $B \in \R$ and $0 < B \leq D$, and $f \in F_k$. Every pole of $f$ within $\Fu$ is captured in $\PF_B f \in F(\infty)$. The order-$M$ pole at $\tau \to \infty$ is captured in $\RF_B f$. The $q$-series coefficients of $\RF_B f$, which we denote, $c_{Rf}(n|B)$, are bounded by 
\begin{align}
\RF_B f(\tau) = \sum_{n = 1}^{\infty} c_{Rf}(n|B) q^n~~,~~C n^{m_f -1} e^{+2 \pi n \beta(B,f)}< |c_{Rf}(n|B)|  < e^{+2 \pi n (\beta(B,f) + \e)} ~, \label{eqCtilde}
\end{align}
where $\e$ is any positive non-zero number, $m_f$ is maximum pole-order of the pole below yet nearest to the line ${\rm Im}(\tau_f) < B$, and $\beta(B,f)$ is the imaginary part of the pole(s) in $f$ with maximal imaginary part that {\em lie below} the strip $V_B$. It is defined by
\begin{align}
\beta(B,f) := {\rm max}\{ {\rm Im}(\tau) \mid \tau \in V~,~ 0 = 1/f(\tau) ~,~\tau \not\in V_B \} .\label{eqBetaData}
\end{align}
Note that $0 < \beta < B$. Thus, there exist $\e>0$ such that $\beta < \e+\beta < B$. 

\subsection{The regularized L-function}\label{secSeriesL}

We now write explicit expressions for $L_f^*(s)^\reg$ when $f \in F_k$. Recall $D$ in Eq.~\eqref{eqDD} and let $B \leq D$. Define $\g_B^{\pm}(\La,\e)$ to be the part of $\g(\La,\e)$ above or below the line ${\rm Im}(\tau) = B$, and define $S\g_B^{+}(\La,\e)$ to be the (orientation-reversed) S-image of $\g_B^{-}(\La,\e)$. Hence,

\begin{lem}
\label{LemReg8}
Because $f(-1/\tau) = \tau^k f(\tau)$, $L_f^*(s|\La,\e)^\reg$ can be rewritten as
\begin{align}
L_f^*(s|\La,\e)^\reg
=& 
\frac12\sum_{\pm \e} c_f(0) \int_{\g_B^{-}(\La,\e)} \frac{d\tau}{\tau} \left( \left( \frac{\tau}{i} \right)^{s} - \left(\frac{\tau}{i} \right)^{k-s} \right) \label{eqLemReg8a}
\\ \nonumber
+&
\frac12\sum_{\pm \e} \left( \int_{\g_B^{+}(\La,\e)} \frac{d\tau}{\tau} \left( \frac{\tau}{i} \right)^{s} \big( f(\tau) -c_f(0) \big) + \int_{S\g_B^{+}(\La,\e)} \frac{d\tau}{\tau} \left( \frac{\tau}{i} \right)^{k-s} i^k\big( f(\tau) -c_f(0) \big) \right)\!.
\end{align}
\end{lem}

\begin{proof}
We prove the Lemma by identical steps to those in the proof of Lemma~\ref{LemReg2}.
\end{proof}

\begin{lem}
\label{LemReg9}
If ${\rm Re}(s) > 0$ and ${\rm Re}(s-k) > 0$, then independent of $\e$ it follows that
\begin{align}
\label{eqLemReg9}
\lim_{\La \to \infty} \frac12\sum_{\pm \e} c_f(0) \int_{\g_B^{-}(\La,\e)} \frac{d\tau}{\tau} \left( \left( \frac{\tau}{i} \right)^{s} - i^k \left(\frac{\tau}{i} \right)^{k-s} \right) =-c_f(0) \bigg(  \frac{B^s}{s} + i^k\frac{B^{s-k}}{k-s} \bigg)~.
\end{align}
\end{lem}
\begin{proof}
This identity follows exactly as in the proofs of Lemma~\ref{LemReg1} and Lemma~\ref{LemReg2}.
\end{proof}
Let $U$, $D$ and $\beta$ be as defined in Eqs.~\eqref{eqDD},~\eqref{eqDU}, and~\eqref{eqBetaData} and choose $B<D$ such that it is not equal to the imaginary part of {\em any} pole of $f$. This is important in Lemma~\ref{LemRegX3}, where we explicitly integrate $\tau^{s-1} {\rm Li}_{-N}(e(\tau-\tau_p))$ along the contours $\g_B^+(\La,\e)$ and $S\g_B^-(\La,\e)$.

After evaluating the constant term in Lemma~\ref{LemReg9}, the remaining part of $L_f^*(s|\La,\e)^\reg$ is a nontrivial integral over contours contained within $V_B$. We partition this nontrivial integral into a part that is regular within the strip $V_B$ (but contains poles at $\infty$) and a part which contains all of the poles of $f$ within $V_B$. Our partitioning involves the following integrals:
\begin{align}
\begin{split}
L_{Rf}^*(s|\La,\e|B)^\reg:= &
\int_{\g_B^+(\La,\e)} \frac{d\tau}{\tau}\left(\frac{\tau}{i}\right)^s (\RF_B f(\tau) - c_f(0)) \\
&\quad +
\int_{S\g_B^+(\La,\e)} \frac{d\tau}{\tau}\left(\frac{\tau}{i}\right)^{k-s}  i^k  (\RF_B f(\tau) - c_f(0)) ~,
\end{split}\label{eqLemRegXa}\\
\begin{split}
L_{Pf}^*(s|\La,\e|B)^\reg:= &
\int_{\g_B^+(\La,\e)} \frac{d\tau}{\tau}\left(\frac{\tau}{i}\right)^s (\PF_B f(\tau)) +
\int_{S\g_B^+(\La,\e)} \frac{d\tau}{\tau}\left(\frac{\tau}{i}\right)^{k-s}  i^k  (\PF_B f(\tau)) 
\end{split}
\label{eqLemRegXb}~.
\end{align}
Note: while $L_f^*(s|\La,\e)$ does not depend on $B$, intermediate terms do depend on $B$.

\begin{lem}
\label{LemRegX}
Let $f \in F_k$ be regular at cusps. For every $\La >U$, and $B \leq D$, the nontrivial integral in $L_f^*(s|\La,\e)$ is equal to the sum $\sum_{\pm \e} \tfrac12 \big(L_{Rf}^*(s|\La,\pm\e|B) + L_{Pf}^*(s|\La,\pm\e|B) \big)$.
\end{lem}

\begin{proof}
For any $X$, $\RF_X f(\tau) = (1-\PF_X) f(\tau)$. Thus, the integrands and the contours of these convergent integrals match. 
\end{proof}

\begin{lem}\label{LemRegX1}
Let $f \in F_k$ be regular at cusps, and let $\La > U$. Then,
\begin{align}
&L_{Rf}^*(s|\La,\e|B)^\reg = \int_{B}^{\La} \frac{dt}{t} t^{s} (\RF_B f(i t)-c_f(0)) + \int_{1/B}^{\La} \frac{dt}{t} t^{k-s}~i^k~  (\RF_B f(i t)-c_f(0)) 
\label{eqLemRegX1a}
\\
&\lim_{\La \to \infty}\! L_{Rf}^*(s|\La,\e|B)^{\reg} = \sum_{n \neq 0} c_{Rf}(n|B) \bigg( \frac{\G(s,2\pi n B)}{(2\pi n)^s}  + i^k \frac{\G(k-s,2\pi n/B)}{(2\pi n)^{k-s}}  \bigg)
\label{eqLemRegX1b}
\end{align}
where $c_{Rf}(n|B)$ are the $q$-series coefficients of $\RF_B f(\tau)$, whose growth is bounded by $|c_{Rf}(n|B)|e^{2\pi (\e-B)n} \to 0$ for some $\e >0$. This sum converges absolutely.
\end{lem}

\begin{proof}
To show Eq.~\eqref{eqLemRegX1a}, note that $\RF_B f(\tau)$ explicitly lacks poles in $V_B$, that $L_{Rf}^*(s|\La,\e|B)$ depends only on the end-points of the contour, and that all of the relevant contours are contained within $V_B$. So we may replace the contours with contours along the vertical axis. 

To show Eq.~\eqref{eqLemRegX1b}, we proceed as in the proof of Lemma~\ref{LemReg2} and commute the $q$-series sum with the integral at finite $\La$. This is possible because the integrals converge at finite $\La$. Recall that by Lemmas~\ref{LemRegC1},~\ref{LemRegC2} and~\ref{LemRegC3} we know that there exists an $\e > 0$ such that $q$-series coefficients $c_{Rf}(n|B)$ of $\RF_B f(\tau)$ are bounded by $e^{-2\pi n B} |c_{Rf}(n|B)| \leq e^{-2 \pi n \e}$. Thus,
\begin{align}
\int_{B}^{\La} \frac{dt}{t} t^{s} (\RF_B f(i t)-c_f(0)) 
&= \int_{B}^{\La} \frac{dt}{t} t^{s} \sum_{n \neq 0} c_{Rf}(n|B) e^{-2\pi n t B} = \sum_{n \neq 0} c_{Rf}(n|B) \int_{B}^{\La} \frac{dt}{t} t^{s} e^{-2\pi n t B} \nonumber\\
&
= \sum_{n \neq 0} c_{Rf}(n|B) \bigg(\frac{\G(s,2\pi n B)}{(2\pi n)^s}-\frac{\G(s,2\pi n \La)}{(2\pi n)^s}\bigg)~.
\end{align}
Because $\G(s,2\pi n) \sim e^{-2 \pi n} (2 \pi n)^{s-1}$, $\G(s,2\pi n) |c_{Rf}(n|B)|$ decays as $(e^{-2 \pi \e})^n(2 \pi n)^{s-1}$ for $0<\e \leq B-\beta$, where $B - \beta>0$ by Eq.~\eqref{eqBetaData}. The sum-representation of the other integral, $\int d\tau i^k (\tau/i)^{k-s-1} (\RF_B f(\tau) - c_f(0))$, converges because $\G(k-s,2\pi n/B)|c_{Rf}(n|B)| \sim e^{-2 \pi n \e}(2 \pi n)^{k-s-1}$. Thus, the two sums exponentially converge for any $s \in \C$. Finally, noting $c_{Rf}(n|B) \G(\sigma,2 \pi n \La)/(2 \pi n)^{\sigma} \to 0$ as $\La \to \infty$ for finite $\sigma \in \C$ proves Eq.~\eqref{eqLemRegX1b}.
\end{proof}

\begin{lem}\label{LemRegX2}
Let $f \in F_k$, and let $\La > U$. Then
\begin{align}
L_{Pf}^*(s|\La,\e|B)^\reg 
&= \sum_{\tau_P \in W_B} \sum_{m > 0} \sum_{\pm \e}\frac{1}{2} \int_{\g_B^{+}(\Lambda,\pm\e)} r^*_{f,\tau_P}(m)  \frac{d\tau}{\tau} \left( \frac{\tau}{i}\right)^s {\rm Li}_{1-m}\big(e(\tau-\tau_P)\big) \label{eqLemRegX2}
\\
&+ \sum_{\tau_P \in W_B} \sum_{m > 0} \sum_{\pm \e}\frac{1}{2} \int_{\g_B^{+}(\Lambda,\pm\e)} r^*_{f,\tau_P}(m)  \frac{d\tau}{\tau} \left( \frac{\tau}{i}\right)^{k-s}~ i^k~ {\rm Li}_{1-m}\big(e(\tau-\tau_P)\big). \nonumber
\end{align}
\end{lem}

\begin{proof}
The proof follows directly from the definition of the function $\PF_B f(\tau) \in F(\infty)$.
\end{proof}

Explicit results for $L_f^*(s)^\reg$ come from evaluating the integral of $(\tau/i)^{\sigma-1} {\rm Li}_{-N}(e(\tau-\tau_p))$ along the contour $\g_{(S)B}^+(\La,\e)$ as $\La \to \infty$ and then as $\e \to 0$, for $\sigma = s$ and $\sigma = k-s$. This is unambiguous when $\tau_p$ is off the imaginary-$\tau$ axis. However, when $f \in F_k$ has poles at imaginary values, $\tau_p = i y$, we must be more careful. Here we begin by defining
\begin{align}
I_N(s,y|\La,\e|B):= \sum_{\pm \e} \frac{1}{2} \int_{\g_B^+(\La,\e)} \frac{d\tau}{\tau} \left( \frac{\tau}{i} \right)^{s-1}{\rm Li}_{-N}(e(\tau-i y))~.
\end{align}
This is unambiguous and finite for any value of $\La \gg y$ and for any value of $\e\neq 0$. 

By holomorphy, $I_N(y,s|\La,\e|B)$ depends only on the sign of $\pm\e$. When $t \gg y$, $|{\rm Li}_{-N}(e^{2\pi (y-t+i\e})| = {\cal O}(e^{-2 \pi t})$, and we may safely take $\La \to \infty$. The resulting finite integral, $I_N(s,y|\e|B):= \lim_{\La \to \infty} I_N(s,y|\La,\e|B)$, depends only on the sign of $\e$. More precisely,
\begin{align}
I_N(s,y|+\e|B) - I_N(s,y|-\e|B) = {\rm Res}(N,s,i y) ~,
\end{align}
where ${\rm Res}(N, s, i y)$ is the residue of $-i (\tau/i)^{s-1} {\rm Li}_{-N}(e(iy-\tau))$. Explicitly, it is given by:
\begin{align}
{\rm Res}(N,s,z) = 
\oint \frac{d\tau}{\tau}  \left(\frac{\tau}{i}\right)^{s}{\rm Li}_{-N}(e(\tau-z))= \BM s-1 \\ N \EM \frac{(-\tau_P)^s~2 \pi~N!}{(2 \pi i \tau_P)^{N+1}} ~.\label{eqLemRegX7}
\end{align}
Note that each integral $\sum_{\pm \e}I_N(s,y|\pm\e|B)$ is separately finite. We now define,
\begin{align}
I_N(s,y|B) :=\lim_{\e \to 0} \sum_{\pm \e} \frac{1}{2} I_N(s,y|\pm \e|B)~.
\end{align}
We may have well-defined expressions for the integrals in Eq.~\eqref{eqLemRegX2} for any $\tau_P \in V_B$. 

\begin{lem}\label{LemRegX3}
Let  ${\rm Im}(\tau_p) = y \neq B$. Then $I_N(s,-i\tau_p|B)$ as defined above equals
\begin{align}
\frac{(-1)}{(2\pi)^N}\sum_{n = 1}^{\infty} \frac{ \ell_n(-i \tau_p|B) }{ (2 \pi n)^{s-N} }
~,~
\ell_n(-i\tau_p|B) = 
\begin{cases}
~~e^{-\pi i (1+2 n \tau_p)} ~~ \G(s,+2\pi n B), &{\rm Im}(\tau_p) < B,
\\
e^{+\pi i (s-N+2 n \tau_p)} ~ \G(s, -2\pi n B), &{\rm Im}(\tau_p) > B.
\end{cases}
\label{eqLemRegX3}
\end{align}
\end{lem}

\begin{proof}
The essence of the proof comes from evaluating $I_N(s,-i \tau_p|D)$ when the pole $\tau_p = i y$ is on the ${\rm Im}(\tau)$-axis. There are only two cases to consider: when $y < B$ and when $y > B$.

When $y < B$, there is no pole along the ${\rm Im}(\tau)$-axis and we may write,
\begin{align}
\int_{B}^{\infty} \frac{dt}{t} t^s~{\rm Li}_{-N}(e^{2 \pi(y-t)}) &= 
\int_{t = B}^{t = \infty} \frac{dt}{t} t^s~\sum_{n = 1}^{\infty} \frac{e^{2 \pi n(y-t)} }{n^{-N}} = \frac{1}{(2\pi)^N}\sum_{n = 1}^{\infty} \frac{e^{2 \pi n y}~\G(s,2\pi n B) }{(2 \pi n)^{s-N}},
\end{align}
which converges because $e^{2 \pi n y}\G(s,2 \pi n B) \sim e^{2 \pi n(y-B)}/(2\pi n)^{s-1}$ exponentially decays.

When $y > B$, then we may break-up the contour into three pieces $\g_B^+(\La,\pm \e):=\g_1 \cup \g^{\pm}_2 \cup \g_3$, where $\g_1:=[B , y-\e]$, $\g_2^{\pm} := \{ y - \e e^{\pm i \theta} \mid \theta \in (0,\pi) \}$ and $\g_3 := [y+\e,\La)$, in the limit where $\La \to \infty$. As emphasized above, we define $I_N(s,y|B)$ as the average of the contour integrals along $\g_B^+(\La,\pm\e)$, which are each separately finite and well-defined. Now, each contour circles the pole at $\tau = i y$ by an angle of $\pm \pi$. Averaging over contours $\g_B^+(\La,\pm \e)$ thus counts the residue at $\tau = i y$ exactly $+\tfrac12-\tfrac12 = 0$ times. Because of this, we can safely evaluate $I_N(s,y|B)$ by considering the integral evaluated at the end-points $t = B$ and at $t = \infty$. Further, because ${\rm Li}_{-N}(e^{-2 \pi t}) = {\cal O}(e^{-2\pi t})$ as $t \to \infty$, we know that the integral evaluated at the upper bound of the integration contour vanishes. 

So we are left with the problem of evaluating the integral at the lower end-point $t = B$. To compute this, we note that when $N$ is a {\em positive} integer, we have
\begin{align}
{\rm Li}_{-N}(x) = (-1)^{N+1} {\rm Li}_{-N}(1/x)~.
\end{align}
This allows us to rewrite the integral evaluated at the lower end-point $t = B$ as,
\begin{align}
\int_{t = B}\frac{dt}{t} t^s {\rm Li}_{-N}(e^{2 \pi (y-t)})
&= (-1)^{N+1}\int_{t = B}\frac{dt}{t} t^s {\rm Li}_{-N}(e^{2 \pi (t-y)})
\\
&= (-1)^{N+1}\int_{t = B}\frac{dt}{t} t^s \sum_{n = 1}^{\infty} \frac{(e^{2 \pi (t-y)})^n}{n^{-N}}
\\
&= (-1)^{N+1}\sum_{n = 1}^{\infty} \frac{e^{-2 \pi n y}}{n^{-N}} \int_{t = B}\frac{dt}{t} t^s e^{2 \pi n t} 
\\
&= (-1)^{N+1}\sum_{n = 1}^{\infty} \frac{e^{-2 \pi n y}}{n^{-N}} \frac{1}{(-2\pi n)^s} \int_{T = -2 \pi n B}
\frac{dT}{T} T^s e^{-T} 
\\
&= \sum_{n = 1}^{\infty} \frac{e^{-2 \pi n y} }{n^{-N}} 
\frac{\G(s,-2\pi n B)}{(-2 \pi n)^{s}} = \sum_{n = 1}^{\infty} \frac{e^{-2 \pi n y} }{n^{-N}} 
\frac{\G(s,-2\pi n B)}{(2 \pi n)^{s}}~.
\end{align}
Thus, because $y > B$, we have $e^{-2\pi n y} |\G(s,-2\pi n B)| \sim e^{-2 \pi n (y-B)}/(2 \pi n)^{s-1}$ for $n \gg 1$. Thus this sum converges exponentially quickly for generic finite $s \in \C$. We finish the proof by noting that the above manipulations apply equally well to $\tau_p = i y + x$ when $x \neq 0$.
\end{proof}

\begin{thm}\label{ThmLfn}
Let $f \in F_k$ be regular at cusps, and let $B$, $c_{Rf}(n|B)$, $\G(s,x)$, $\tau_p$, $V_B$, $r^*_{f,\tau_P}(m)$, and $I_N(s,z|A)$ as defined above. Then the L-integral that yields the L-function is
\begin{align} 
L_f^*(s)^{\reg} =&-c_f(0) \bigg(\frac{B^s}{s} + \frac{i^k~B^{s-k}}{k-s} \bigg) + \sum_{n \neq 0} c_{Rf}(n|B) \bigg( \frac{\G(s,2\pi n B)}{(2\pi n)^s}  + i^k \frac{\G(k-s,2\pi n/B)}{(2\pi n)^{k-s}}  \bigg) \nonumber\\
&+\sum_{\tau_p \in V_B} \sum_{m > 0} r^*_{f,\tau_P}(m) \bigg( \! I(m-1,s,-i\tau_p|B) + i^k I(m-1,k-s,-i\tau_p|1/B) \! \bigg).\!
\label{eqLfn}
\end{align}
\end{thm}

\begin{proof}
In Lemma~\ref{LemReg4}, we defined $L_f^*(s)^\reg$ as the $\La \to \infty$ and $\e \to 0$ limit of the average of the integrals over contours that begin at $\tau = i/\La$ and end at $\tau = i \La$ and pass just $\e$ the left, and just $\e$ the right, of poles on the imaginary-$\tau$ axis. Each integral converges, and yields a finite result at fixed $\La>U$ and $\e\neq 0$, and each has a finite limit as $\La \to \infty$ and $\e \to 0$.

Then in Lemma~\ref{LemReg8}, we split-up each of individual contour into a sum of three terms: two contour integrals with integrands $(\tau/i)^{\sigma-1}(f(\tau) -c_f(0))$, and two contour integrals with integrands $(\tau/i)^{\sigma -1} c_f(0)$, where $\sigma = s$ and $\sigma = k-s$. In Lemma~\ref{LemReg9}, we showed that the integral over $(\tau/i)^{\sigma -1} c_f(0)$ yields the polar terms in $L_f^*(s)^\reg$ given by $-c_f(0)(\tfrac{B^s}{s} + i^k \tfrac{B^{s-k}}{k-s})$. 

Following this, in Lemma~\ref{LemRegX}, we split-up the integral over the nontrivial contours, which are entirely contained in the strip $V_B \subset V$ into two pieces. One piece is regular within $V_B$ but contains poles at cusps whose two integrands are proportional to $\RF_B f(\tau) - c_f(0)$. The other piece is regular at the cusp but has poles within $V_B$ whose two integrands are proportional to  $\PF_B f(\tau)$. We call these terms $L_{Rf}^*(s|\La,\e|B)^\reg$ and $L_{Pf}^*(s|\La,\e|B)^\reg$, respectively.

By appealing to the $q$-series growth bounds in Lemmas~\ref{LemRegC1},~\ref{LemRegC2}, and~\ref{LemRegC3} and Proposition~\ref{PropGrowthCap} and to Lemma~\ref{LemReg2}, we showed that the $\RF_B f(\tau) - c_f(0)$ integrals equal the sums
\begin{align}
\sum_{n \neq 0} c_{Rf}(n|B) \left(\frac{\G(s,2\pi n B)}{(2\pi n)^s}+i^k \frac{\G(k-s,2\pi n /B)}{(2\pi n)^{k-s}}\right)~, 
\end{align}
which converge exponentially quickly for finite $s \in \C$. In Lemmas~\ref{LemRegX1},~\ref{LemRegX2}, and~\ref{LemRegX3} we found explicit expressions for $I_N(-i\tau_p,s|B)$, the integrals within $\PF_B f(\tau)$. 

Combining these results yields Eq.~\eqref{eqLfn} and completes the proof of the Theorem.
\end{proof}

This has an important Corollary:

\begin{cor}
\label{CorLfn1}
If $L_f^*(s)^\reg$ is as defined Theorem~\ref{ThmLfn}, then $L_f^*(s)^\reg= i^k L_f^*(k-s)^\reg$.
\end{cor}

\begin{proof}
To prove this, it is instructive to first restrict our attention to the case where $f \in M_k$. Here, $L_f^*(s|\La,\e)^\reg$ is  manifestly independent of $B$~\cite{04-BFK}: when $f \in M_k$, $B$ simply amounts to a turning-point where the contour $\g(\La,\e)$ is inflected to the contour $\g_B(\La,\e):= \g_B^+(\La,\e) \cup S\g_B^+(\La,\e)$. Thus, when evaluating $L_f^*(s|\La,\e)^\reg$ we inflect around $B$, and when evaluating $L_f^*(k-s|\La,\e)^\reg$ we inflect around $1/B$. The results for $L_f^*(s)^\reg$ and $L_f^*(k-s)^\reg$ are each given wholly by terms in the first line of Eq.~\eqref{eqLfn}, and are identical up to the overall factor $i^k$. (See similar analysis in~\cite{04-BFK} for $L_f^*(s)^\reg$ when $f \in M_k^!$.)

This procedure holds for $f \in F_k$: we define $L_f^*(s|\La,\e)^\reg$ by inflecting around ${\rm Im}(\tau) = B$, and writing (the nontrivial part of) $L_f^*(s|\La,\e)^\reg$ in terms of contour integrals along $\g_B(\La,\e)= \g_B^+(\La,\e) \cup S\g_B^+(\La,\e)$. We then split-up the integral into the sum of two integrals, whose integrands are $(\tau/i)^{s-1} (\RF_B f(\tau) - c_f(0))$ and $(\tau/i)^{s-1} \PF_B f(\tau)$. These integrands are {\em defined} by the property that all poles ``below'' $\tau = i \infty$ and above the line at ${\rm Im}(\tau) = B$ are projected out of $\RF_B f$ and are entirely contained in $\PF_B f$. We call the corresponding integrated expressions, respectively, $L_{Rf}^*(s|\La,\e|B)^\reg$ and $L_{Pf}^*(s|\La,\e|B)^\reg$.

The key point here is that regardless of reflecting about the line ${\rm Im}(\tau) = 1/B$ or ${\rm Im}(\tau) = B$, to obtain convergent expressions for the integrals $L_{Rf}^*(s|\La,\e|B)^\reg$ and $L_{Pf}^*(s|\La,\e|B)^\reg$ we must project-out all poles above the {\em lowest} of the two lines ${\rm Im}(\tau) = B$ or ${\rm Im}(\tau) = 1/B$. 

We now obtain $L_f^*(k-s|\La,\e)^\reg$ in a form that directly compares with $L_f^*(s)^\reg$ in the following way. We preform the same procedure on $L_f^*(k-s|\La,\e)^\reg$, save that we now reflect the contour $\g(\La,\e)$ about the line ${\rm Re}(\tau) = 1/B$. Now, we project-out all of the poles of $f$ that lie below the {\em lowest} of the two lines ${\rm Im}(\tau) = B$ or ${\rm Im}(\tau) = 1/B$, to obtain the two integrands $(\tau/i)^{k-s-1} (\RF_B f(\tau) -c_f(0))$ and $(\tau/i)^{k-s-1} \PF_B f(\tau)$. These integrands match those in $L_f^*(s|\La,\e)$ above, save with $s \to k-s$ and the inflection-point reversed $B \to 1/B$. Evaluating the integral for $L_f^*(k-s|\La,\e)$ in this way yields expressions identical to those in Eq.~\eqref{eqLfn}, save with $s \to k-s$ and $B \to 1/B$, just as for the above special case where $f \in M_k^!$. Similarly, just as for the special case where $f \in M_k^!$, each pair of terms within $(L_f^*(s)^\reg,L_f^*(k-s)^\reg)$ exactly match, save with a relative factor of $i^k$ between them. 
\end{proof}

\begin{cor}\label{CorLfn2}
Suppose $f \in F_k$. According to the definition of $L_f^\reg(s)$ in Theorem~\ref{ThmLfn},
\begin{align}
\lim_{s \to 0} L_f^\reg(s) &= -c(0)~ (1-\delta_{k,0})~,~\label{s=0a} \\
\lim_{s \to k} L_f^\reg(s) &= 
\begin{cases}
~0\qquad \qquad,~&k < 0 \\ 
~{\rm divergent}~~ ,~ &k > 0 
\end{cases}
~,~\label{s=ka} \\
\lim_{s \to N} L_f^\reg(s) &= 
\begin{cases} 
0~~,\qquad ~&k \leq 0~,~N \in \{ k+1, \ldots, -1\} \\ 
{\rm finite}~~,~ ~  &k > 0~,~N \in \{ 1, \ldots, k-1\} ~~\end{cases} ~.~\label{s=Na}
\end{align}
\end{cor}

\begin{proof}
Special values of $L_f^*(s)^\reg$ within the critical strip, $s \in [0,k]$, for $f \in F_k$ are given by Eq.~\eqref{eqLfn}. When $f \in F_0$, then by Theorem~\ref{ThmLfn} the residues of the polar terms $\tfrac{B^s}{s} + i^k \tfrac{B^{s-k}}{k-s}$ in $L_f^*(s)^\reg$ cancel exactly. Thus, Eqs.~\eqref{s=0a},~\eqref{s=ka}, and~\eqref{s=Na} collapse to a single evaluation,
\begin{align}
\lim_{s \to 0} \frac{(2\pi)^s}{\G(s)} L_f^*(s)^\reg\bigg|_{k =0} = \lim_{s \to 0} \frac{(2 \pi)^s}{\G(s)} \left\{ -c(0) \bigg(\frac{B^s}{s} + i^0\frac{B^{s-0}}{0-s} \bigg) + {\cal O}(s^0)  \right\} = 0~.\label{Cor1Eq3a}
\end{align}
However when $f \in F_k$ for $k \neq 0$, the two polar terms do not cancel, thus:
\begin{align}
\lim_{s \to 0} \frac{(2\pi)^s}{\G(s)} L_f^*(s)^\reg\bigg|_{k \neq 0} = \lim_{s \to 0} \frac{(2 \pi)^s}{\G(s)} \left\{ -c(0) \frac{B^s}{s} + {\cal O}(s^0)  \right\} = -c(0)~.\label{Cor1Eq3b}
\end{align}
Together, Eqs.~\eqref{Cor1Eq3a} and~\eqref{Cor1Eq3b} establish the claim in Eq.~\eqref{s=0a} for every $k$, and establish Eqs.~\eqref{s=ka} and~\eqref{s=Na} for $k = 0$. We now show Eqs.~\eqref{s=ka} and~\eqref{s=Na} for $k \neq 0$.

We now consider $k < 0$. Here, integers in the critical strip are $N \in \{ k, k+1, \ldots, -1,0\}$. Importantly, $\G(s)$ has a pole at each of these critical values of $s$. Therefore, 
\begin{align}
\lim_{s \to k} \frac{(2\pi)^s}{\G(s)} L_f^*(s)^\reg\bigg|_{k < 0} 
&\!\!\!= \lim_{s \to k} \frac{(2 \pi)^s}{\G(s)} \left\{ -c(0) \frac{i^k~B^{s-k}}{k-s} + {\cal O}((s-k)^0)  \right\} = 0~,\label{Cor1Eq4a} \\
\lim_{s \to N} \frac{(2\pi)^s}{\G(s)} L_f^*(s)^\reg\bigg|_{k < 0} 
&\!\!\!= \lim_{s \to N} \frac{(2 \pi)^s}{\G(s)} \left\{ {\cal O}((s-N)^0)  \right\} = 0~.\label{Cor1Eq4b}
\end{align}
We now consider $k > 0$. Here, the integers in the critical strip are $N \in \{ 0, 1, \ldots, k-1,k\}$. Importantly, $\G(s)$ is regular at each of these critical values of $s$. So,
\begin{align}
\lim_{s \to k} \frac{(2\pi)^s}{\G(s)} L_f^*(s)^\reg\bigg|_{k > 0} 
& = \lim_{s \to k} \frac{(2 \pi)^s}{\G(s)} \left\{ -c(0) \frac{i^k~B^{s-k}}{k-s} + {\cal O}((s-k)^0)  \right\} = {\rm divergent}~,\label{Cor1Eq4a} \\
\lim_{s \to N} \frac{(2\pi)^s}{\G(s)} L_f^*(s)^\reg\bigg|_{k > 0} 
&\!\!\!= \lim_{s \to N} \frac{(2 \pi)^s}{\G(s)} \left\{ {\cal O}((s-N)^0)  \right\} = {\rm finite}~.\label{Cor1Eq4b} 
\end{align} 
This shows Eqs.~\eqref{s=ka} and~\eqref{s=Na} for $k > 0$ and finishes the proof. \end{proof}

\subsection{Numerics, polar Rademacher sums, path dependence, and poles at cusps}\label{sec15-HR1919}

In this section, we briefly comment on several aspects of our definition of L-functions for meromorphic modular forms, before moving on to applications. First, we discuss writing meromorphic modular forms in terms of a Rademacher sum over polar terms and the related issue of numerical evaluations of our L-functions. Second, we discuss ambiguities in our L-functions which arise when poles cross the integration contour, and whether this affects the special values of $L_f^\reg(s)$ in Theorem~\ref{ThmLfn} and Corollary~\ref{CorLfn2}.

The pole-subtraction procedure used to define $\RF_B f$ and $\PF_B f$ is extremely simple when $f \in F_k$ has negative-definite weight, $k < 0$. When $k < 0$, then meromorphic modular forms can be written as a convergent sum over non-positive weight polylogarithms that have poles at all distinct modular images of all of the poles of $f$. For example, as Hardy and Ramanujan showed~\cite{15-HR1919} that $1/E_6(\tau)$ can be written in the following way:
\begin{align}
\frac{1}{E_6(\tau)} 
= \sum_{\g \in \G/\G_{\infty}} \frac{{\rm Res}(i)}{(c (i) + d)^{8}} {\rm Li}_0\left(e\left(\tau-\frac{a (i) + b}{c (i) + d}\right)\right) 
= \sum_{n = 0}^{\infty} D(n) q^n~,
\label{eqRad1}
\end{align}
where $\G := \SL_2(\Z)$, and ${\rm Res}(i)$ is the residue of $1/E_6(\tau)$ on the pole at $\tau = i$. Crucially, the $q$-series coefficients for $1/E_6(\tau)$ increase exponentially, as
\begin{align}
D(n) 
= \sum_{(c,d)=1} \frac{{\rm Res}(i)}{(c+di)^8} ~ e\left(-\frac{a i + b}{c i + d} \right) 
= {\rm Res}(i) \sum_{\lambda(c,d)} \frac{D_{\lambda(c,d)}(n)}{\lambda(c,d)^4} ~ e^{2 \pi n/\lambda(c,d)}~,
\label{eqRad2}
\end{align}
where the first sum is over all coprime integers $(c,d) = 1$, and the second sum is over the related quantity $\lambda = c^2 + d^2$. For a fixed coprime pair $(c,d)$ there is a pair of elements of $\G/\G_{\infty}$ given by $\g_{\pm} = \BSM \pm a & b \\ c & \pm d \ESM$, which yield $D_{\lambda(c,d)}(n)$:
\begin{align}
D_{\lambda(c,d)}(n) := \frac{1}{2} \sum_{\pm} \left(\frac{c \mp d i}{c \pm d i}\right)^4~e\left( {\rm Im}(\g_{\pm}(i)) \right) ~.
\label{eqRad3}
\end{align}
More generally, when $f \in F_k$ has a single simple pole within the fundamental domain, $\tau_p \in {\cal F}$, then it can be written similarly as a kind of Rademacher sum over ${\rm Li}_0(e(\tau - \g \tau_p))$:
\begin{align}
f(\tau) = \sum_{\g \in \G/\G_{\infty}} \frac{{\rm Res}(\tau_p)}{(c \tau_p + d)^{2-k}} {\rm Li}_0(e(\tau-\tfrac{a \tau_p + b}{c \tau_p + d}))
\label{eqRad4}
\end{align}
This has been widely extended to meromorphic modular forms for $k < 0$ in~\cite{16-BialekPol1, 17-BialekPol2, 18-BialekPol3, 19-KolnPol1, 20-KolnPol2} that have multiple poles within ${\cal F}$ of arbitrary order. 

In this context, the pole-subtraction procedure used to define $\RF_B f$ and $\PF_B f$ is particularly simple to implement. When the general results of~\cite{19-KolnPol1, 20-KolnPol2} apply, we obtain rapidly convergent expressions for $c_{Rf}(n|B)$ by simply deleting terms from poles above the ${\rm Im}(\tau) = B$ line. Our ability to numerically evaluate $L_f^*(s)^\reg$, as defined in Theorem~\ref{ThmLfn}, appear to depend crucially on having Rademacher-like sums for meromorphic modular forms $f \in F_k$.

However, when $k = 0$, then sums akin to Eq.~\eqref{eqRad4} do not converge absolutely and do not reproduce meromorphic modular functions without modification. When $k > 0$ then these sum no longer converge absolutely; when $k \geq 1$ they do not converge at all. In the absence of convergent Rademacher-sums over polar $q$-series terms, such as ${\rm Li}_{1-p}(e(\tau - \g\tau_p))$, for $k \geq 0$, then it is less obvious how to write convergent expressions for the $c_{Rf}(n|D)$ in the pole-subtracted sums that we take to define the regularized L-function for generic meromorphic modular forms. We know of one example where this has been explicitly done for $k > 0$: When $k = 2$, Bringmann et al~\cite{21-DivisorsMF1} build on the work of Bruinier et al~\cite{22-DivisorsMF2} and explicitly construct sums of this sort for weight-two quasimodular forms, given by $\d_\tau \log f$ for some modular form $f$. We expect numerically evaluate our expressions for $L^*_f(s)$ when $f \in F_k$ to be crucially tied to the existence of convergent sums such as in~\cite{21-DivisorsMF1} when $f \in F_k$ for $k \geq 0$. 

We now briefly discuss the path-(in)dependence of our result for $L_f^*(s)^\reg$ when $f \in F_k$. Our definition of $L_f^*(s)^\reg$ hinged on writing $f \in F_k$ as the sum $f = \RF_B f + \PF_B f$. Here $\RF_B f$ is regular in the vertical strip $V_B$, and $\PF_B f$ contains every pole of $f$ in $V_B$. We then integrated $f = \RF_B f + \PF_B f$ along the ${\rm Im}(\tau)$-axis, from $\tau = i 0^+$ to $\tau = i \infty$. Now, precisely because $f$ has poles within ${\cal F}$, this definition is path-dependent. 

Had we defined the L-function by a different path $\g$, denoted $L_f^*(s)^\reg_{\g}$, we would have
\begin{align}
L_f^*(s)^\reg_{\g} = \int_{\g} \frac{d\tau}{\tau} \left( \frac{\tau}{i} \right)^s (f(\tau) - c_f(0)) = L_f^*(s)^\reg + 2\pi i \sum_{i = 1}^N \chi(\tau_i) {\rm Res}(m_i,s,\tau_i)~,
\end{align}
where $m_i$ is the pole-order, ${\rm Res}(m_i,s,\tau_i)$ is the residue at $\tau_i$ in Eq,~\eqref{eqLemRegX7}, and $\chi(\tau_i) = \pm1$ tracks if the pole $\tau_i$, which lies between $\g$ and the ${\rm Im}(\tau)$-axis, is to the right or left of the axis. In this sense, the L-functions depend explicitly on path. 

Similarly, consider a continuous family of meromorphic modular forms indexed by the location of one of their poles, $z$. When we smoothly send $z$ to $\g z$ for $\g \in \SL_2(\Z)$, the contour integral that defines $L_f^*(s)^\reg$ will acquire residues corresponding to how many times images of the pole at $z$ cross the vertical ${\rm Im}(\tau)$-axis. These additional pieces resemble the discontinuities in theta-decompositions of meromorphic Jacobi forms in~\cite{23-Zwegers, 24-DMZ}, and could thus also represent interesting wall-crossing phenomena. However, any concrete connection to wall-crossing would be highly premature at this stage.

Importantly, these additional contributions do not have any poles at finite values of $s \in \C$. Hence, they do not change the residues at the poles $1/s$ and $1/(k-s)$ in $L_f^*(s)^\reg$. Therefore, the special values discussed in Corollary~\ref{CorLfn2} are insensitive to this ambiguity. Because the main results in sections~\ref{secSV} and~\ref{secQFT} that we derive from our L-functions for $f \in F_k$ come from these special values, they are not sensitive to this ambiguity. 

Finally, we reiterate that $L_f^*(s)^\reg$ applies to $f \in F_k$ that are regular at cusps. This is {\em complimentary} to, but does not directly extend or generalize, the results in~\cite{04-BFK, 05-BFK2} where they define $L_f^*(s)^\reg$ for $f \in M_k^!$, which have poles exclusively at cusps. However, as commented in the proof of Lemma~\ref{LemReg2}, up to a finite number of divergent terms in the infinite sum in Eq.~\eqref{eqHouston}, the results in Eqs.~\eqref{eqLemReg2} and~\eqref{eqHouston} exactly their results. It will be important to fuse these two results, and define regularized L-integrals and L-functions for all $f \in F_k$.

\section{Special values of regularized L-functions}\label{secSV}

This section has two principal components. The first part is in sections~\ref{svTSM1},~\ref{svTSM2}, and~\ref{svTSM3}. The second part is in section~\ref{svAltSum}. Throughout this section, $\G := \SL_2(\Z)$.

In sections~\ref{svTSM1}--\ref{svTSM3} we present one of the main results of this paper, which concerns L-functions for weight-two meromorphic modular forms $\Lambda_d \in F_2$,
\begin{align}
\Lambda_d(\tau) = H(d) + \sum_{n = 1}^{\infty} t_n(d) q^n \in F_2~, \label{s3e1}
\end{align}
where $H(d)$ is the Hurwitz class-number, and that counts the number of distinct quadratics with negative discriminant $-d = 4ac - b^2 > 0$, and $t_n(d)$ is the trace of the unique modular function $J_n(\tau) = 1/q^n + {\cal O}(q) \in M_0^{!}$ over CM-points of a quadratic with negative discriminant $-d$, and $d \equiv 0, 3$ (mod 4). Evaluating the L-function of $\Lambda_d(\tau)$ at $s = 0$ gives a sum-rule that relates the traces over singular moduli and the Hurwitz class numbers,
\begin{align}
\lim_{s \to 0} L_{\Lambda_d}^\reg(s) = -H(d) = \reg\sum_{n =1}^{\infty} t_n(d)~,\label{s3e2}
\end{align}
where ``$\reg$'' denotes regularization via L-function. This relationship between the Hurwitz class numbers and traces of singular moduli is new: it rests upon the L-functions for meromorphic modular forms defined in section~\ref{secLfunction}. Later on, in section~\ref{secQFT}, we relate the special values in sections~\ref{svTSM1} and~\ref{svTSM2} to novel statements about Casimir energies, central charges, and a new reflection-symmetry, for special two-dimensional conformal field theories (CFTs). Finally, in section~\ref{svTSM3}, we show an amusing way to calculate Hurwitz class-numbers via bounding the exponential growth of traces of singular moduli as a function of $d$.  

Then, in section~\ref{svAltSum}, we study a subset of meromorphic modular forms that vanish identically as $\tau \to i 0^+$. We relate these results to the special values in Corollary~\ref{CorLfn2} of Theorem~\ref{ThmLfn}, at negative integer values of $s$ and at $s = 0$. Each result has the interpretation as the sum of exponentially diverging sequences of numbers. Where they overlap, they agree perfectly. We view this as a check on the consistency $L_f^\reg(s)$ when $f \in F_k$.

\subsection{Traces of singular moduli and the Hurwitz class numbers}\label{svTSM1}

One of our main applications of our L-functions concerns a new relationship between the number of $\G$-inequivalent quadratics with a fixed negative discriminant, $-d$, and the regularized sum of traces of modular functions,
\begin{align}
J_m(\tau) = \frac{1}{q^m} + \sum_{n = 1}^{\infty} c_m(n) q^n \in M_0^{!} ~, \label{eqJm}
\end{align}
when traced over solutions to $\G$-inequivalent quadratics with negative discriminant $-d$. 

To begin, we define $-d$ to be the negative discriminant of a quadratic operator, $Q(X,Y) = a X^2 + b XY + c Y^2$, where $a,b,$ and $c$ are all integers. Here, $-d= 4 a c - b^2$. Further, we define the {\em CM point} $\alpha_Q$ to be a solution of a given quadratic with negative discriminant $-d$ such that
\begin{align}
\alpha_Q:= \frac{-b \pm i \sqrt{d}}{2a}~~,~~Q(1,\alpha_Q) = 0~~,~Q \in {\cal Q}_d/\G~~,\alpha_Q \in {\cal F}~. \label{eqAlphaQ}
\end{align} 
In Eq.~\eqref{eqAlphaQ}, we called the space of quadratics with negative discriminant $-d$ by the name ${\cal Q}_d$. We now associate the number $w_Q$ to each distinct quadratic with discriminant $-d$,
\begin{align} 
w_Q:= 
\begin{cases}
3~, & Q(X,Y) = a(X^2 + XY + Y^2)~,\\ 
2~, & Q(X,Y) = a(X^2 + Y^2)~,\\
1~, & {\rm otherwise}~.
\end{cases}
\label{eqWq}
\end{align}
Summing over all $w_Q$ for the finite sum $Q \in {\cal Q}_d/\G$ gives Hurwitz class numbers in Eq.~\eqref{s3e1}:
\begin{align}
H(d) := \sum_{Q \in {\cal Q}_d/\G} \frac{1}{w_Q}~.\label{eqHd}
\end{align}
Weighting each term in the sum by the value of $J_m(\tau)$ at the unique root/CM-point $\alpha_Q \in {\cal F}$ of the quadratic $Q \in {\cal Q}_d/\G$, then we obtain the traces of singular moduli $t_n(d)$:
\begin{align}
t_n(d) :=  \sum_{Q \in {\cal Q}_d/\G} \frac{1}{w_Q}J_n(\alpha_Q)~.\label{eqTrD}
\end{align}
Finally, we define the $d^{th}$ Hilbert class polynomial as,
\begin{align}
{\cal H}_d(X) := \prod_{Q \in {\cal Q}_d/\G} (X - j(\alpha_Q))^{1/w_Q}~,\label{eqHxD}
\end{align}
where $j(\tau):= J_1(\tau) + 744$. It is straightforward to show that,
\begin{align}
{\cal H}_d(j(\tau)) &= q^{-H(d)} (1- t_1(d) q+ {\cal O}(q^2))
\end{align}
In \cite{02-TSM}, Zagier noted that this follows directly from the definitions and noted that
\begin{align}
\Lambda_d(\tau) 
  := \frac{1}{2 \pi i } \frac{d}{d\tau} \log\big({\cal H}_d(j(\tau)) \big) 
~ = \frac{1}{2 \pi i } \sum_{Q \in {\cal Q}_d/\G} \frac{1}{w_Q} \frac{j'(\tau)}{j(\tau) - j(\alpha_Q)} \in F_2~, \label{eqLambdaD2}
\end{align}
is the generating function of the traces of $J_n(\tau)$, from Eq.~\eqref{eqJm}, when summed over all unique roots of distinct quadratics with negative discriminant $-d$ from Eq.~\eqref{eqAlphaQ}. Concretely,
\begin{align}
\Lambda_d(\tau)
= \sum_{Q \in {\cal Q}_d/\G} \frac{1}{w_Q} \bigg(1 + \sum_{n = 1}^{\infty} J_n(\alpha_Q) q^n \bigg)
= H(d) + \sum_{n = 1}^{\infty} t_n(d) q^n~. \label{eqLambdaD}
\end{align}
Further, in \cite{02-TSM}, Zagier proved that the Hilbert class polynomials ${\cal H}_d(j(\tau))$ span a subspace of meromorphic modular forms that have Borcherds product expansions~\cite{25-Bor1994}. The connection to Borcherds products is a crucial aspect of \cite{02-TSM} and of this paper. We will revisit it first in section~\ref{svTSM2}, and then throughout section~\ref{secQFT}. However, our current focus is a precise statement of the special values of the L-function of the generating function for the traces of singular moduli, $L_{\Lambda_d}^\reg(s)$.

Using Theorem~\ref{ThmLfn} and Corollary~\ref{CorLfn2} from section~\ref{secLfunction}, we find the L-function of $\Lambda_d(\tau)$:
\begin{align}
L_{\Lambda_d}^\reg(s) = \frac{(2 \pi)^s}{\G(s)} \left( \frac{H(d)}{2-s}- \frac{H(d)}{s} + {\rm regular}(s) \right) ~~,~~{\rm for} ~~d \neq 0~.
\end{align}
Here ``${\rm regular}(s)$'' refers to the terms in $L_f^*(s)^\reg$ that are finite for finite $s \in \C$ (see Theorem~\ref{ThmLfn} for details). Using this L-function, we see immediately 
\begin{align}
\lim_{s \to 0} L_{\Lambda_d}^\reg(s) &= \lim_{s \to 0} \frac{(2 \pi)^s}{\G(s)}\bigg( -\frac{H(d)}{s} + {\cal O}(s^0) \bigg) = -H(d)~,\label{eqSum1a}\\
\lim_{s \to 1} L_{\Lambda_d}^\reg(s) &= \lim_{s \to 1} \frac{(2 \pi)^s}{\G(s)}\bigg( {\cal O}((s-1)^0) \bigg) = {\cal O}(1) < \infty~. \label{eqSum1b}
\end{align}
Now, if we interpret the special value of $L_{\Lambda_d}^\reg(s)$ at $s = 0$ from Eq.~\eqref{eqSum1a} as a sum of traces of singular moduli of the $J_m(\tau)$ at imaginary quadratic points $\alpha_Q$ with discriminant $-d$, 
\begin{align}
\sum_{Q \in {\cal Q}_d/\G} \frac{1}{w_Q} \bigg( 1 + \reg\sum_{n = 1}^{\infty} J_n(\alpha_Q)\bigg) = \sum_{Q \in {\cal Q}_d/\G} \frac{1}{w_Q} \big( 0 \big) = 0~,\label{eqCool1}
\end{align}
then it is a formal relationship between traces of singular moduli and Hurwitz class numbers. 

This relationship between traces of singular moduli and the $H(d)$ is new. It is important to note that the location of the simple poles in $\Lambda_d(\tau)$ in Eq.~\eqref{eqLambdaD} at the points $\alpha_Q$ in Eq.~\eqref{eqAlphaQ}, combined with Lemmas~\ref{LemRegC2} and~\ref{LemRegC3}, together imply the refined information that
\begin{align}
|J_n(\alpha_Q) - (e^{-2 \pi i \alpha_Q})^n| = |J_n(\alpha_Q) - e^{+\pi n \sqrt{d}}(-1)^{n b/2a}| < e^{2 \pi n({\rm Im} (\g_{\rm max} \cdot \alpha_Q)+\e)}~, \label{eqJdBound}
\end{align} 
where $\g_{\rm max} \cdot \alpha_Q$ is the first modular image of $\alpha_Q$ that is distinct from $\alpha_Q$, and  $\e \geq0 $. 

As ${\rm Im}(\g_{\rm max} \cdot \alpha_Q) < 1/\sqrt{2}$, these growth conditions must be less than $e^{\pi n \sqrt{2}}$. (See Theorem~\ref{HurwitzThm} in section~\ref{svTSM3} for details.) Without L-functions for meromorphic modular forms, there would be no meaning to this sum of exponentially growing $q$-series coefficients of $\Lambda_d \in F_2$. We explore several of consequences of this in section~\ref{svTSM2} and again in section~\ref{secQFT}.

We now emphasize that the vanishing in Eq.~\eqref{eqCool1} happens term-by-term in the sum over $Q \in {\cal Q}_d/\G$. If we define,
\begin{align}
\La_{\alpha_Q}(\tau):= \frac{1}{w_Q} \left(\frac{1}{2\pi i} \frac{j'(\tau)}{j(\tau) - j(\alpha_Q)}\right)~,
\end{align}
then linearity of the integral $L_{f+g}^*(s)^\reg = L_{f}^*(s)^\reg + L_{g}^*(s)^\reg$, implies
\begin{align}
L_{\Lambda_d}^\reg(s) = \sum_{Q \in {\cal Q}_d/\G} L_{\Lambda_{\alpha_Q}}^{\reg}(s)~. \label{eqSumL}
\end{align}
Hence, each $L_{\La_{\alpha_Q}}^\reg(s)$ in the sum over $Q \in {\cal Q}_d/\G$ in Eq.~\eqref{eqSumL} individually vanishes as $s \to 0$.

Before moving on to Borcherds products, we give a mnemonic for understanding the exponential growth of the $q$-series coefficients of $\La_d(\tau)$. By Eq.~\eqref{eqLambdaD}, we have $c_{\La_d}(n) = t_n(d)$. Now, $t_n(d)$ is ``mostly'' $J_n(i \sqrt{d}/2+x)$, where $x = 0$ or $1/2$ depending on $d$. Further, $J_n(\tau)$ is ``mostly'' $e^{-2 \pi i n \tau}$. Thus the dominant contribution to $|t_n(d)|$ is ``mostly'' $e^{+n \pi \sqrt{d}}$.

\subsection{Traces of singular moduli and Borcherds products}\label{svTSM2}

Traces of singular moduli, and the related Hurwitz class numbers, are famously related to Borcherds infinite product formulas for modular forms. Using this, we reinterpret the special values of $L_{\Lambda_d}^\reg(s)$ in section~\ref{svTSM1}. This reinterpretation will be of considerable interest in physical applications, when the quantum field theoretic path integral has an infinite product expansion (with integer exponents). In \cite{02-TSM} Zagier proved that the modular polynomials ${\cal H}_d(j(\tau))$ in Eq.~\eqref{eqHd} are an important class of Borcherds products. Recall that $d = 4 ac - b^2 \equiv 0,3$ (mod 4). Thus, 
\begin{align}
{\cal H}_d(j(\tau)) = q^{-H(d)} \prod_{n = 1}^{\infty} (1-q^n)^{A(n^2,d)} = \prod_{Q \in {\cal Q}_d/\G} (j(\tau)-j(\alpha_Q))^{1/w_Q}~, \label{eqBZ1}
\end{align} 
where $A(n^2,d)$ can either be interpreted as $c(n^2)$ in the $q$-series of a particular modular form $f_d(\tau)$ that is an element of the Kohnen plus-space of $M_{1/2}^{!}(\G_0(4))$, or as $c(d)$ in the $q$-series of the related particular modular form $g_{n^2}(\tau) \in M_{3/2}^{!}(\G_0(4))$. 

The modular forms $f_d(\tau)$ and $g_{n^2}(\tau)$ are defined by their modular weight in the subgroup $\G_0(4)$ and their $q$-series growth:
\begin{align}
\begin{cases}
~f_d(\tau) := q^{-d} + \sum_{n = 1}^{\infty} A(n,d) q^n \in M_{1/2}^{!}(\G_0(4))~, \\
g_{n^2}(\tau) := q^{-n^2} - 2- \sum_{d = 1}^{\infty} A(n^2,d) q^d \in M_{3/2}^{!}(\G_0(4))~. \label{eqFdGn2}
\end{cases}
\end{align}
The $g_{n^2}(\tau)$ in Eq.~\eqref{eqFdGn2} are a subset of a larger set of weight-$3/2$ forms,
\begin{align}
g_D(\tau) := \frac{1}{q^D} - 2 S(D)- \sum_{d = 1}^{\infty} A(D,d) q^d \in M_{3/2}^{!}(\G_0(4))~, \label{eqGd}
\end{align}
where for $D > 0$, we have $S(D) = 1$ when $D$ is a perfect square, and $S(D) = 0$ otherwise.

The $f_d(\tau)$ in Eq.~\eqref{eqFdGn2} and the $g_d(\tau)$ in Eq.~\eqref{eqGd} considered in~\cite{02-TSM} have $d > 0$. Further, the $d = 0$ cases of each have prominent locations in classical number theory:
\begin{align}
f_0(\tau) :&= 1 + \sum_{n = 1}^{\infty} A(n,0) q^n =1 + \sum_{n \in \Z} q^{n^2} = \theta(\tau) \in M_{1/2}(\G_0(4))~, \label{eqF0} \\
g_0(\tau) :&= -\frac{1}{12} + \sum_{d = 1}^{\infty} A(0,d) q^d = H(0) + \sum_{d = 1}^{\infty} H(d) q^d = \mathscr{H}(\tau) \in \widetilde{M}_{3/2}(\G_0(4))~, \label{eqG0}
\end{align}
where $H(0) = -1/12$ and $\mathscr{H}(\tau)$ is the holomorphic part of Zagier's weight-3/2 mock modular form, $\hat{G}(\tau)$. In this context, we see that the $q$-series coefficients of the $\Lambda_d(\tau)$s can be recast in terms of the $q$-series coefficients of $g_D(\tau)$ or equivalently of $f_d(\tau)$:
\begin{align}
\Lambda_d(\tau) =  \frac{1}{2 \pi i}\frac{d}{d\tau}\log\big( {\cal H}_d(j(\tau)) = A(0,d) + \sum_{n = 1}^{\infty} \bigg( \sum_{m|n} m A(m^2,d) \bigg) q^n \in F_2 ~, \label{eqLambdaD3}
\end{align}
where we have written $A(0,d)$ in lieu of $H(d)$. In our context, the special value in Eq.~\eqref{eqSum1a} equates to the formal sum-rule,
\begin{align}
H(d) + \lim_{s \to 0} L_{\Lambda_d}^\reg(s) 
= A(0,d) + \lim_{s \to 0} \sum_{n = 1}^{\infty} \bigg( \sum_{m|n} m A(m^2,d)\bigg) n^{-s} 
= 0~. \label{eqFactorizeLf0}
\end{align}
This formal L-function is structurally identical to a convolution of Dirichlet series. If we de-convolve this formal expression, then we arrive at the formal identity,
\begin{align}
\sum_{n = 1}^{\infty} \bigg( \sum_{m|n} \frac{m A(m^2,d)}{n^{s}}\bigg) = 
\bigg(\sum_{n = 1}^{\infty} \frac{A(n^2,d)}{n^{s-1}} \bigg) \bigg( \sum_{m = 1}^{\infty} \frac{1}{m^s} \bigg)~. \label{eqFactorizeLf}
\end{align}
Noting that the second factor in Eq.~\eqref{eqFactorizeLf} is simply $\zeta(s)$, we can reinterpret the two special values in Eqs.~\eqref{eqSum1a} and~\eqref{eqSum1b} as
\begin{align}
\lim_{s \to 0} L_{\Lambda_d}^\reg(s) &= \lim_{s \to 0} \zeta(s) \bigg(\sum_{n = 1}^{\infty} \frac{A(n^2,d)}{n^{s-1}} \bigg) = -A(0,d)~, \label{eqSum2a} \\
\lim_{s \to 1} L_{\Lambda_d}^\reg(s) &= \lim_{s \to 1} \zeta(s) \bigg(\sum_{n = 1}^{\infty} \frac{A(n^2,d)}{n^{s-1}} \bigg) = {\rm finite} < \infty~. \label{eqSum2b}
\end{align}
Removing the factor of $\zeta(s)$ implies the two formal sum-rules,
\begin{align}
\lim_{s \to -1} \bigg(\sum_{n = 1}^{\infty} \frac{A(n^2,d)}{n^{s}} \bigg) = +2 A(0,d) \quad , \quad  
\lim_{s \to 0} \bigg(\sum_{n = 1}^{\infty} \frac{A(n^2,d)}{n^{s}} \bigg) = 0 \quad . \label{eqSumFull}
\end{align}
This is a new relationship between the $q$-series coefficients of the various $g_D(\tau)$. 

This relationship will be of special use and importance in section~\ref{secQFT} where we study CFTs whose path integrals have Borcherds products. As we show in section~\ref{secSTL}, Casimir energies for these CFTs are directly analogous to the Hurwitz class numbers, $H(d)$, from Eqs.~\eqref{eqSum1a} and~\eqref{eqSum2a}. In section~\ref{secTrex} we comment that {\em both} sum-rules in Eq.~\eqref{eqSumFull} were anticipated by a recently observed symmetry of path integrals in quantum field theory \cite{06T-rex1, 07T-rex2, 14T-rex0}.

Before concluding, we emphasize that de-convolving the formal Dirichlet series in Eq.~\eqref{eqFactorizeLf0} into the product of formal Dirichlet series in Eq.~\eqref{eqFactorizeLf} is well-motivated on physical grounds. We discuss the physical basis for this factorization in section~\ref{secSTL}.

\subsection{Using Borcherds exponents to count quadratics with discriminant $-d$}\label{svTSM3}

Poles above the ${\rm Im}(\tau) = \sqrt{3}/2$-line correspond to $q$-series growth in excess of $e^{\pi n \sqrt{3}}$, as proven in Lemmas~\ref{LemRegC2} and~\ref{LemRegC3}. Bounding the growth conditions for the Borcherds exponents, $A(n^2,d)$, is thus directly sensitive to the unique roots of quadratics with negative discriminant $-d$. Put differently, bounding the $A(n^2,d)$ directly counts the Hurwitz class numbers, $H(d) = A(0,d)$. We emphasize this fact, even though it lies somewhat outside the main scope of the paper. It is most efficient to state this result as a Theorem:

\begin{thm}\label{HurwitzThm}
Consider the two expressions for $\Lambda_d \in F_2$ defined in Eqs.~\eqref{eqLambdaD2} and~\eqref{eqLambdaD}. The $q$-series coefficients $\tilde{t}_n(d)$ of the pole-subtracted function,
\begin{align}
\tilde{\Lambda}_d(\tau) = \Lambda_d(\tau) - \sum_{Q \in {\cal Q}_d/\G} \frac{1}{w_Q}{\rm Li}_0(e^{2 \pi i (\tau-\alpha_Q)}) = \sum_{n = 0}^{\infty} \tilde{t}_n(d) q^n~, \label{eqThm3}
\end{align}
have exponential growth that is bounded by $\tilde{t}_n(d) < e^{\pi n\sqrt{2}}$. 
\end{thm}

\begin{proof}
To begin the proof, we note that the pole-subtraction procedure removes all poles from $\Lambda_d(\tau)$ that lie either within, or on the boundary of, the fundamental domain. However, $\Lambda_d(\tau)$ has an infinite number of poles within the strip $|{\rm Re}(\tau)| \leq 1/2$ that come from the modular orbits of the finite set of points $\alpha_Q$. It is straightforward to show that the maximum imaginary part of any image of any $\alpha_Q$ under $\G$, which is neither within the fundamental domain nor its boundary, is ${\rm Im}(\frac{A  \alpha_Q + B}{C  \alpha_Q + D}) \leq \sqrt{2}/2$. This bounds ${\rm Im}(\g_{\rm max} \cdot \alpha_Q) <\sqrt{2}/2$. 

To prove the Theorem, we overlay the two expressions for $\Lambda_d(\tau)$ in Eqs.~\eqref{eqLambdaD2} and~\eqref{eqLambdaD}:
\begin{align}
\Lambda_d(\tau) 
&= \frac{1}{2 \pi i} \sum_{Q \in {\cal Q}_d/\G} \frac{1}{w_Q} \frac{j'(\tau)}{j(\tau) - j(\alpha_Q)}
= H(d) + \sum_{n = 1}^{\infty} t_d(n) q^n~.
\end{align}
Now, $\Lambda_d(\tau)$ has a unique simple pole with residue $1/w_Q$ for every point $\alpha_Q$. Because $\Lambda_d \in F_2$, it has poles with unit residues at all modular images of every distinct $\alpha_Q$. 

Every $\alpha_Q$ is within the fundamental domain, and hence has ${\rm Im}(\alpha_Q) > 1/\sqrt{2}$. By Lemmas~\ref{LemRegC2} and~\ref{LemRegC3}, then the pole-subtracted function,
\begin{align}
\tilde{\Lambda}_d(\tau) = \Lambda_d(\tau) - \sum_{Q \in {\cal Q}_d/\G} \frac{1}{w_Q}{\rm Li}_0(e^{2 \pi i (\tau-\alpha_Q)})
\end{align}
is regular at all points within the fundamental domain, and at all points on its boundary. So any pole in $\tilde{\Lambda}_d(\tau)$ has ${\rm Im}( \tau) < 1/\sqrt{2}$. Thus the only source of exponential growth in the $q$-series coefficients of $\tilde{\Lambda}_d$ can come from poles below this line. Thus, $|\tilde{J}_d(n)| < e^{\pi n\sqrt{2}}$.
\end{proof}

Theorem~\ref{HurwitzThm} gives an alternative way to determine the Hurwitz class-numbers, directly from the $A(n,d)$: by cancelling the leading growth of the $q$-series coefficients of $\Lambda_d(\tau)$. Subtract-off exponentially growing contributions to the $A(n^2,d)$ for any given fixed-$d$, until growth is bounded by ${\rm Exp}[\pi n \sqrt{2}]$. The number of terms which must be subtracted before hitting this bound gives the number of quadratics with negative discriminant $-d$.

\subsection{A consistency condition for $L_f^\reg(s)$}\label{svAltSum}

L-functions associated with modular forms represent a regularization of the formally divergent sum, $\sum_n c(n)/n^s$. In this section, we derive related results for the behavior of meromorphic modular forms, $f = \sum_n c(n) q^n$, in the limit where $q$ goes to one. To this end, we prove Lemmas~\ref{thm:nSumCP} and~\ref{thm:nSumN}, which agree with the results from Corollary~\ref{CorLfn2}. Specifically, Lemma~\ref{thm:nSumCP} applies to the case when $f \in F_k$ vanishes cusps and has arbitrary weight, while Lemma~\ref{thm:nSumN} applies to the case when $f \in F_k$ is bounded at cusps and has negative modular weight. 

We begin by defining the space of quasimodular meromorphic forms of weight-$(k+2\Delta)$ given by the $\Delta$-fold $\tau$-derivative acting on elements of $F_k$. We denote this space by $\tilde{F}_{(k,2 \Delta)}$:
\begin{align}
\tilde{F}_{(k,2\Delta)} := \bigg\{~ \left(\frac{1}{2 \pi i}\frac{d}{d \tau}\right)^\Delta f(\tau) \bigg| ~f \in F_k ~\bigg\}~. \label{eqTildeF1}
\end{align}
Further, we define the $\Delta^{th}$ descendant of a specific $f \in F_k$ by $f^{(\Delta)} \in \tilde{F}_{(k,2\Delta)}$,
\begin{align}
f^{(\Delta)}(\tau) 
:&= \left(\frac{1}{2 \pi i} \frac{d}{d\tau}\right)^\Delta f(\tau) 
= \sum_n n^\Delta c(n) q^n 
= \sum_{m = 0}^{\Delta} g_m(\tau) ~ E_2(\tau)^m~,  \label{eqTildeF2}
\end{align}
where $g_m(\tau) \in F_{k+2(\Delta-m)}$ and $E_2(\tau)$ is the weight-two quasimodular holomorphic Eisenstein series. We now state the first Lemma:

\begin{lem}\label{thm:nSumCP}
Suppose $f \in F_k$, and $f(\tau) = \sum_{n \geq 1} c(n) q^n$. Then for any $k$ and for every non-negative integer $\Delta \geq 0$, it follows that 
\begin{align}
\lim_{\tau \to 0} f^{(\Delta)}(\tau) = \lim_{q \to 1} \sum_{n = 1}^{\infty} n^{\Delta} c(n) q^n =0 ~\label{s=NbCP}.
\end{align}
\end{lem}

\begin{proof}
First, consider $\Delta = 0$. Because $f = \sum_{n \geq 1} c(n) q^n$, we have
\begin{align}
\lim_{\tau \to i \infty} f(\tau) = \lim_{t \to \infty} \sum_{n = 1}^{\infty} c(n) e^{-2 \pi n t} = c(1) \lim_{t \to \infty} e^{-2\pi t} \big( 1 + {\cal O}(e^{-2\pi t}) \big)~.\label{eqStuff1}
\end{align}
Because $f \in F_k$, we can relate behavior at $\tau \to i \infty$ to behavior at $\tau \to i 0^+$, and have
\begin{align}
\lim_{\tau \to i 0^+} f(\tau) = \lim_{t \to \infty} (i t)^{k} \sum_{n = 1}^{\infty} c(n) e^{-2 \pi n t} = i^k c(1) \lim_{t \to \infty} t^k e^{-2\pi t} \big( 1 + {\cal O}(e^{-2\pi t}) \big) = 0~.\label{eqStuff2}
\end{align}
A crucial point here is that the decay in~\eqref{eqStuff1} is exponential rather than power-law. Thus, $f(\tau)$ decays exponentially both when $\tau \to i \infty$ and when $\tau \to i 0^+$, regardless of its weight $k$. 

Second, we consider $\Delta > 0$. When $\Delta = 1$, it is straightforward to see $f^{(1)}(\tau) = g_0 + E_2 g_1 $, where $g_0 \in F_{k+2}$ and $g_1 = \tfrac{k}{12} f \in F_k$. Because $f$ vanishes at the cusp, we see that the product $E_2 f$ also vanishes at cusps. Therefore, so too must $g_0$. Iterating this implies that when $f \in F_k$ vanishes at the cusp, then every $\tau$-derivative of $f$ has powers of $E_2(\tau)$ multiplied by modular forms that vanish as $e^{2 \pi n i \tau}$ as $\tau \to i \infty$. Thus, if $f = {\cal O}(q)$ and $f \in F_k$, then
\begin{align}
f^{(\Delta)}(\tau) = \sum_{m = 0}^{\Delta} g_m(\tau) E_2(\tau)^{\Delta-m} ~~,~~g_m(\tau) = \sum_{n = 1}^{\infty} d_m(n) q^n \in F_{k + 2m}~~. \label{eQquasiMF1}
\end{align}
This holds for $\Delta \geq 1$. As $\tau \to i 0^+$, $f^{(\Delta)}$ vanishes as $e^{-2\pi i/\tau}/\tau^{k+2\Delta}$. This 
follows from the fact that $f^{(\Delta)}$ decays as $e^{-2 \pi |\tau|}$ when $\tau$ goes to $\infty$. As $g_m \in F_{k + 2m}$ vanishes at $q = 0$, then 
\begin{align}
\lim_{\tau \to i \infty} g_m(\tau) E_2(\tau)^{\Delta-m} 
= \lim_{t \to i \infty} d_m(1) e^{-2 \pi n t} \big( 1 + {\cal O}(e^{-2\pi t}) \big) 
= \lim_{t \to i \infty} d_m(1) e^{-2 \pi n t}~. \label{eQquasiMF2}
\end{align}
Each term $g_m(\tau) E_2(\tau)^m$ the expansion of $f^{(\Delta)}$ in Eq.~\eqref{eQquasiMF1} exponentially decays when $\tau \to i \infty$ from Eq.~\eqref{eQquasiMF2} combines with the quasimodularity of $E_2(\tau)$ to imply
\begin{align}
\lim_{\tau \to i 0^+} f^{(\Delta)}(\tau) 
&= \lim_{\tau \to i 0^+} \sum_{m = 0}^{\Delta} g_m(\tau) E_2(\tau)^{\Delta-m} 
= \sum_{m = 0}^{\Delta}  \lim_{\tau \to i \infty} g_m\left(-\frac{1}{\tau}\right) E_2\left(-\frac{1}{\tau}\right)^{\Delta-m} \\
&= \sum_{m = 0}^{\Delta}  \lim_{\tau \to i \infty} \left( \tau^{k + 2m}g_m\left(\tau\right) \right) ~\left( \tau^2 E_2\left( \tau\right) + \frac{6 \tau}{\pi i} \right)^{\Delta-m} 
\sim \lim_{t \to \infty} t^{k + 2 \Delta} e^{-2 \pi t}~,\!\!\!\!
\end{align}
which vanishes exponentially quickly for any integer $\Delta \geq 0$. This completes the proof. \end{proof}

Before stating the next Lemma, we pause to focus on a concrete example. Consider $f = 1/E_4 \in F_{-4}$. Explicit computation shows that for $f^{(\Delta)}$, none of the $g_m(\tau)$s will vanish as $q^{x}$ ($x > 0$) at cusps if $\Delta \leq 4$. Yet, $f^{(5)}$ is a meromorphic form in $F_{+6}$ that vanishes as $q^{x}$ ($x > 0$) at cusps. This is a consequence of Bol's identity~\cite{26-Bol1} (see discussion of Eq. (4.14) in~\cite{27-Bol2}). Thereafter, every $g_m(\tau)$ that occur in any quasimodular polynomial expansion of $f^{(\Delta > -|k|)} \in \tilde{F}_{(-|k|,2 \Delta)}$ will vanish as $q^{x}$ ($x > 0$) at cusps. Lemma~\ref{thm:nSumCP} thus implies vanishing of $f^{(\Delta > -|k|)}(\tau)$ when $\tau \to i 0^+$. So:

\begin{lem}\label{thm:nSumN}
Suppose $f \in F_k$ with $k \leq 0$, and that $f(\tau) = \sum_{n \geq 0} c(n) q^n$. Then for every non-negative integer $0 \leq \Delta \leq |k|/2$, and then for every $\Delta \geq |k|+1$, it follows that 
\begin{align}
\lim_{q \to 1} \sum_{n = 1}^{\infty} n^{\Delta} c(n) q^n = +c(0) \delta_{k,0} \delta_{\Delta,0}~\label{s=NbN}.
\end{align}
\end{lem}

\begin{proof}
First, we consider $k \leq 0$ and $\Delta = 0$. By assumption, we have,
\begin{align}
\lim_{\tau \to i \infty} f(\tau) = c(0) + \lim_{t \to \infty} c(1) e^{-2 \pi t} (1 + {\cal O}(e^{-2\pi t)})~.
\end{align}
Thus, because $k \leq 0$ we have,
\begin{align}
\lim_{\tau \to i 0^+} f(\tau) = \lim_{\tau \to i \infty} \tau^{k} f(\tau) = \lim_{t \to \infty} t^{-|k|} \big( c(0) + c(1) e^{-2 \pi t} +  {\cal O}(e^{-2\pi t)} \big) = c(0)\delta_{k,0}~.
\end{align}

Second, we consider $k \leq 0$ and $0< \Delta < |k|/2$. For every $\Delta$ in this range, $f^{(\Delta)}$ is a quasimodular form of negative weight without a constant term. Thus, we may use modularity to relate the exponential decay of $f^{(\Delta)}$ as $\tau \to i \infty$ to an exponential decay of $f^{(\Delta)}$ as $\tau \to i 0^+$. 

Crucially, each term in the polynomial expansion of $f^{(\Delta)}$ in terms of $E_2(\tau)^m g_m(\tau)$ has negative weight when $\Delta < |k|/2$. Inspecting examples, however, shows that while the constant term of $f^{(\Delta)}(\tau)$ vanishes for $\Delta > 0$, the individual $g_m \in F_{k + 2 (\Delta - m)}$ do not vanish at cusps. And so the only way to ensure that the individual terms, $E_2^{\Delta - m} g_m$, do not diverge as $\tau \to i 0^+$ is to constrain their quasimodular weight to be negative.

Now, consider the boundary case $\Delta = |k|/2$. Here, $f^{(\Delta)} = f^{(|k|/2)}$ is a meromorphic quasimodular form of weight-zero whose $q$-series that is ${\cal O}(q)$. Hence, the leading polynomial divergences from the $g_m E_2^m$ terms cancel when $\tau \to i 0^+$:
\begin{align}
\lim_{\tau \to i 0^+} f^{(|k|/2)}(\tau) 
&= \lim_{\tau \to i \infty} f^{(|k|/2)}\left(-\frac{1}{\tau} \right) 
=  \lim_{\tau \to i \infty} \sum_{m = 0}^{|k|/2} g_m\left(-\frac{1}{\tau} \right) E_2\left(-\frac{1}{\tau} \right)^m \\
&=  \lim_{\tau \to i \infty} \sum_{m = 0}^{|k|/2} \tau^{2m-2m} g_m\left(\tau \right) \left( E_2\left(\tau \right) + \frac{6}{\pi i \tau} \right)^m \\
&=  \lim_{\tau \to i \infty} \sum_{m = 0}^{|k|/2} g_m\left(\tau \right) E_2\left(\tau \right)^m + {\cal O}\left(\frac{1}{\tau}  \right)~.\!\!
\end{align}
Because the leading term is ${\cal O}(\tau^0)$ and the sub-leading terms decay as $\tau \to i \infty$ as $e^{-2\pi|\tau|}/\tau$ or faster, we may safely focus on the first term. Now, note that this leading term is simply equal to $f^{(|k|/2)}$ itself. Further, recall that $f^{(|k|/2)}$ does not have a constant term. Thus,
\begin{align} 
\lim_{\tau \to i 0^+} f^{(|k|/2)}(\tau) 
&= \lim_{\tau \to i \infty} \sum_{m = 0}^{|k|/2}  f^{(|k|/2)}(\tau) 
= \lim_{\tau \to i \infty} \sum_{m = 0}^{|k|/2}  c(1) e^{2 \pi i \tau} \big(1 + {\cal O}(e^{2 \pi i \tau}) \big)  = 0~.
\end{align}
This fails when $|k|/2 \leq \Delta \leq |k|$. When $|k|/2 \leq \Delta \leq |k|$, then the corresponding $f^{(\Delta)}$ is a positive weight quasi-modular form. Because the $g_m(\tau)$ have constant terms, and $f^{(\Delta)}$ in this range has positive weight, we conclude that $f^{(\Delta)}$ diverges as $\tau \to i0^+$. 

The situation changes when $\Delta \geq |k| + 1$. Using Bol's identity~\cite{26-Bol1}, we conclude that $f^{(|k|+1)}$ is an exactly modular meromorphic form within $F_{|k|+2} = F_{k+2|k|+2}$ that lacks a constant term. Thus, by Lemma~\ref{thm:nSumCP}, we can conclude that
\begin{align}
\lim_{\tau \to i 0^+} f^{(\Delta)}(\tau) = 0~,
\end{align}
whenever $\Delta \geq |k|+1$. This completes the proof.
\end{proof}

We now recap. In Corollary~\ref{CorLfn2} we derived concrete expressions that stand in for the L-functions of meromorphic modular forms, at $s = 0$ and at negative integer values of $s$. These L-functions formally represent the finite part of the nowhere convergent object, $\sum_n c(n)/n^s$, at these special values. Independently, in Lemmas~\ref{thm:nSumCP} and~\ref{thm:nSumN}, we exploited modularity to explicitly compute the finite limit of meromorphic modular forms, and their $\tau$-derivatives, when $\tau \to i 0^+$. This gives an independent computation of sums of the form, $\sum_n n^{\Delta} c(n) q^n$ in the limit where $q$ goes to one. Within their regions of overlap, these two completely independent methods to regularize the sum of $q$-series coefficients are entirely consistent. This overlap occurs in two situations. 

First, in Lemma~\ref{thm:nSumCP} we study the behavior of meromorphic modular forms $f \in F_k$ for any $k$ that vanish at cusps, and their meromorphic quasi-modular descendants $f^{(\Delta)} \in \tilde{F}_{(k,2 \Delta)}$. We show that $f^{(\Delta)}(\tau)$ vanishes exponentially quickly as $\tau \to i 0^+$ for every integer $\Delta \geq 0$. When written in terms of the $q$-series of $f^{(\Delta)}$, this amounts to the statement that,
\begin{align}
\lim_{q \to 1} \sum_{n = 1}^{\infty} n^{\Delta} c(n) q^n = 0~.
\end{align}
This is consistent with the L-functions of meromorphic modular forms in Corollary~\ref{CorLfn2},
\begin{align}
\lim_{s \to -n} \frac{(2\pi)^s}{\G(s)} L_f^*(s)^\reg = 0~,
\end{align}
when $c(0)$ vanishes.

Second, in Lemma~\ref{thm:nSumN} we study meromorphic modular forms $f \in F_k$ for $k \leq 0$, and their quasi-modular derivatives $f^{(\Delta)} \in \tilde{F}_{(k,2 \Delta)}$, that are bounded at cusps. Specifically, we show that when $\tau$ is sent to $i 0^+$, $f^{(\Delta)} \in \tilde{F}_{(k,2 \Delta)}$ vanish when $\tau \to i 0^+$ for $0 \leq \Delta \leq |k|/2$. These vanishing results match the special values of the L-function of $f \in F_k$ for $k \leq 0$, spelled-out in Corollary~\ref{CorLfn2}, at integer values of $s$ within the critical strip $k \leq s \leq 0$. Further, by Bol's identity~\cite{26-Bol1}, we show that the limit $\tau \to i 0^+$ of $f^{(\Delta)}$ vanishes when $\Delta \geq |k|+1$. This matches the trivial zeros of $L_f^\reg(s)$ at negative integers $s$, enforced by the $1/\G(s)$ factor built into the L-function of $f \in F_k$. Explicitly,
\begin{align}
c(0) + \lim_{s \to 0} L_f^\reg(s) & = c(0) \delta_{k,0} = \lim_{q \to 1} \sum_n n^{+0} c(n) q^n~,~ \\ 
\lim_{s \to -\Delta} L_f^\reg(s) & = 0= \lim_{q \to 1} \sum_n n^{\Delta} c(n) q^n  ~,~~1 \leq \Delta \leq |k|/2 \\
\lim_{s \to -\Delta} L_f^\reg(s) & = 0= \lim_{q \to 1} \sum_n n^{\Delta} c(n) q^n ~,~~\Delta \geq |k|+1 .
\end{align} 

Note that Lemma~\ref{thm:nSumN} evaluates the L-function for negative weight meromorphic modular forms at negative integers {\em in the critical strip}, which lies between $s = 0$ and $s = -|k|$ (recall $k = -|k| < 0)$. As the evaluations in this section rest on a direct appeal to modularity, they are limited to $\tau$-derivatives of the original $f \in F_{-|k|}$ with non-positive weight.  We can appeal to the symmetry $L_f^*(s)^\reg = i^k L_f^*(k-s)^{\reg}$ in Corollary~\ref{CorLfn1}, to show that the other $|k|/2$ special values within the critical strip for negative weight meromorphic modular forms agree with the vanishing results obtained directly from modularity. We leave explicit confirmation of this agreement at $s = k, k+1, \ldots, k/2+1$ ($k < 0$) to future work.

\section{Sum-rules in conformal field theory}\label{secQFT}

In this section, we use the L-functions developed in the previous sections to extract interesting physical data that characterizes conformal field theories in two dimensions (2d CFTs). In particular, we show that the L-function for the logarithmic derivatives of the path integral at $s = 0$ gives the Casimir energies and thus central charges of certain unitary CFTs. We further show that the special value at $s = 1$ of the same L-function confirms a sum-rule that was motivated by a recently noticed symmetry of path integrals in quantum field theory (QFT), when applied to 2d CFT path integrals with infinite product expansions. 

The structure of this section is as follows. In section~\ref{secBasics}, we briefly describe the physical motivation and setting for our study of sum-rules in 2d CFTs. In particular, we emphasize why modularity appears in 2d CFTs~\cite{08-Yellow}, discuss holomorphic factorization, and the role of unique ground-states in rewriting path integrals as infinite products. In section~\ref{secSTL}, we introduce the stress-energy tensor, and prove a theorem about the special values of its L-function when the path integral is a weight-$k$ modular form whose $q$-series begins with $q^{-\Delta}$. Finally, in sections~\ref{secTrex} and~\ref{secEnsemble} we use this L-function to verify the sum-rules suggested by T-reflection, and point-out a class of functions that seem related to traces of singular moduli, but do not seem to often be discussed.

Concretely, this discussion can be brought to bear on 2d CFTs that holomorphically factorize, such as the monster CFT with $Z(\tau) = J(\tau)$~\cite{28-FLM1984} and the conjectural extremal 2d CFTs thought to be dual to Einstein gravity in $AdS_3$~\cite{29-W2007}. Strikingly, these sum-rules substantiate sum-rules formally derived from demanding QFT path integrals be invariant under T-reflections~\cite{06T-rex1, 07T-rex2}. Finally, these sum-rules exactly agree with the recent extension of meromorphic modular forms from the upper half-plane to the double half-plane~\cite{30T-rex2-GL2}.

\subsection{Central charges in 2d CFTs and special values of L-functions}\label{secBasics}

By definition, the path integral for a QFT integrates over all allowed configurations of the fields in the theory over all points on the space-time manifold. Carrying-out the path integral for a 2d CFT placed on the two-torus obliterates all information about the structure of the two-torus, save for the lattice of points identified by the toroidal compactification and periodicity conditions along the non-contractible cycles of the two-torus. 

For this reason, all 2d CFT path integrals on the two-torus, denoted $Z(\tau,\overline{\tau})$, are explicit functions of this lattice, $\Lambda(\tau):=\{ m + n \tau \mid m,n \in \Z\}$. Path integrals defined on two lattices that are equal up to a scale transformation, $\Lambda(\tau) = z \Lambda(\tau')$ for complex $\tau,z,\tau'\neq 0$, must be equal. Thus, all 2d CFT path integrals on the two-torus must be modular invariant. 

The discussion in this section applies to 2d CFT path integrals that factorize holomorphically and have unique ground-states. When the CFT factorizes holomorphically, we have
\begin{align}
Z(\tau, \overline{\tau}) = Z(\tau) \overline{Z}(\overline{\tau})~,
\end{align}
where $Z$ and $\overline{Z}$ correspond to the path integrals for the decoupled left- and right-movers of the CFT. Holomorphic factorization conventionally means that $Z$ and $\overline{Z}$ are separately modular invariant, and thus are modular functions: $Z, \overline{Z} \in M_0^!$. 

However, one of the main motivations of this work is to think of meromorphic modular forms as a possible testing-ground for more general QFT path integrals that behave well under modular transforms and have (Hagedorn) poles. So, in this spirit, we allow ourselves to consider path integrals $Z$ that have nontrivial modular weight: $Z \in F_k$ for $k\in \tfrac12 \Z$. For convenience, we will call $Z$ a CFT path integral even if it has nontrivial weight or poles.

When the CFT has a unique ground-state, then the lowest power of $q$ in its $q$-series has unit coefficient, while the other $q$-series coefficients are all integers:
\begin{align}
Z(\tau) = q^{-\Delta} \sum_{n = 0}^{\infty} D(n) q^n~~,~~D(0) = 1~~,~~D(n) \in \Z~.
\label{eqZsig}
\end{align}
Here $-\Delta$ is called the {\em Casimir energy} of the CFT. It represents the vacuum-energy of the CFT. The $D(n)$ count the number of states in the full CFT with energy $n$ above the vacuum.

As they count states, the $D(n)$ are positive integers. (Twisted indices and path integrals may have $D(n)$ negative.) Because $D(0) = 1$ and $D(n) \in \Z$, we can rewrite $Z(\tau)$ as
\begin{align}
Z(\tau) = q^{-\Delta} \prod_{n = 1}^{\infty} (1-q^n)^{-d(n)} 
\label{eqZpi}
\end{align}
where the $d(n) \in \Z$ for every $n$. (Had $D(0) \neq 1$, then the $d(n)$ fail to be integers.)

Correlation functions of the stress-energy tensor of the theory, $T_{\mu \nu}(\tau)$, carry crucially important information about the theory. We find it very useful to study the expectation-value of the stress-energy tensor for the full CFT, $\langle T(\tau) \rangle$:\footnote{The statement that the one-point function of the stress-energy tensor for the CFT is given by the log-derivative of the path integral can be understood in two ways. First, mechanically the $\partial_{\tau} \log Z$ gives the expectation value of energy in a statistical ensemble. Second, from general principles the stress-energy tensor is sensitive to variations in length-scales in the geometry of the spacetime manifold. Thus, it is natural that $\langle T \rangle$ is given by a functional derivative of the path integral with respect to the shape of the two-torus, $\tau$. Normalizing this derivative by the needed factor of $1/Z$ yields $\langle T \rangle \propto \partial_{\tau} \log Z$, as in Eq.~\eqref{eqTvev}.}
\begin{align}
\langle T(\tau) \rangle:= \frac{1}{2\pi i} \frac{d}{d\tau} \log Z(\tau) ~.
\label{eqTvev}
\end{align}
Crucially, if $Z(\tau)$ vanishes or has some pole at some finite value of $\tau$, then $\langle T(\tau) \rangle$ has a simple pole. For the theories we are interested in, the stress-tensor has the following $q$-series:
\begin{align}
\langle T(\tau) \rangle = -\Delta -\sum_{n = 1}^{\infty} q^n \bigg( \sum_{m|n} m d(m) \bigg)~.
\end{align}
This has a number of important consequences. 

First, by the bounds in section~\ref{secGrowLF}, if $Z(\tau)$ has a finite collection of singular points in the fundamental domain, and if the point with largest imaginary part is at $\tau = z$, then there exists a real number $C > 0$ such that $C e^{2 \pi n {\rm Im}(z)} \leq  |m d(m)| < e^{2 \pi n ({\rm Im}(z) + \e)}$ for every $\e > 0$. (The upper-bound is always satisfied. The lower bound is satisfied for infinitely many positive integers $n$.)

Second, from the structure of Theorem~\ref{ThmLfn}, we know that the special value of $L_{\langle T \rangle}^\reg(s)$ at $s = 0$ is exactly the constant term of $\langle T(\tau) \rangle$: the Casimir energy, $-\Delta$. This happens despite the growth of the $q$-series coefficients. From a physics perspective, this is an amusing result.

In unitary 2d CFTs, the Casimir energy $\Delta$ is directly proportional to the central charge $c$ of the CFT. The central charge is commonly referred to as a measure of the degrees of freedom in the theory. Thus, realizing the Casimir energy as the special value of the L-function for the stress-energy tensor is logically equivalent to stating that the central charge can be literally interpreted as the sum of the number of excitations in the CFT.

This is striking. Even if the CFT is both (a) free CFT and (b) satisfies the rather strong constraint of holomorphic factorization, in general the path integral will vanishes at some value of $\tau$. Hence, by Proposition~\ref{PropGrowthCap}, the $d(n)$ exhibit exponential/Hagedorn growth. To even write-down the L-function for the stress-energy tensor, we need L-functions for meromorphic modular forms in Theorem~\ref{ThmLfn}. In sections~\ref{secSTL} and~\ref{secEnsemble} we develop this physical picture further.

\subsection{The L-function of the one-point function of the stress-energy tensor}\label{secSTL}

In this section, we prove two results about the structure of $L_{\langle T \rangle}^\reg(s)$. (The proofs rest heavily on the beautiful paper~\cite{22-DivisorsMF2}.) Before doing so, we must introduce a bit of language from statistical mechanics. To begin, we argue that the following two functions,
\begin{align}
Z_{\rm GC}(\tau)&:= q^{-\Delta} \sum_{n = 0}^{\infty} D(n)q^n ~, 
\label{eqZgc} 
\\
Z_{\rm CAN}(\tau)&:= d(0) + \sum_{n = 1}^{\infty} d(n) q^n~, 
\label{eqZcan}
\end{align}
play distinct and important roles when the 2d CFT is (in some sense) free. Here, $Z_{\rm CAN}(\tau)$ is the {\em canonical partition function} which counts the number of distinct excitations in the {\em single-particle} Fock-space of the theory with a given energy $n$ above the vacuum, while $Z_{\rm GC}(\tau)$ is the {\em grand canonical partition function} which counts the number of distinct states in the full {\em multi-particle} Fock-space with a given energy $n$ above the vacuum.

The physical justification for naming these two functions is as follows. When the quantum field theory is free, then the generating functions for the single-particle Fock-space $Z_{\rm CAN}(\tau)$ and the generating function for the multi-particle Fock-space $Z_{\rm GC}(\tau)$ are related by the following combinatoric map:
\begin{align}
Z_{\rm GC}(\tau) = q^{-\Delta} {\rm exp}\bigg[- \sum_{m = 1}^{\infty} \frac{1}{m} \bigg( Z_{\rm CAN}(m \tau) - d(0) \bigg) \bigg] = q^{-\Delta} ~\sum_{n = 0}^{\infty} D(n) q^n~. \label{eqCombMap}
\end{align}
With $Z_{\rm GC}(\tau)$ and $Z_{\rm CAN}(\tau)$ as defined in Eqs.~\eqref{eqZgc} and~\eqref{eqZcan}, we see that if $D(0) = 1$ then applying the map in Eq.~\eqref{eqCombMap}, we find:
\begin{align}
Z_{\rm GC}(\tau) = q^{-\Delta} \sum_{n = 0}^{\infty} D(n)q^n = q^{-\Delta} \prod_{n = 1}^{\infty} (1-q^n)^{-d(n)}~.
\end{align}
This is {\em exactly} the infinite product factorization of $Z(\tau)$ for 2d CFTs in Eqs.~\eqref{eqZsig} and~\eqref{eqZpi}.

Famously, particle number is not fixed in QFT: particle production can happen. As path integrals consider all fluctuations for all possible field configurations of a given QFT or CFT, then, the path integral for a CFT or QFT does not naturally map onto the canonical partition function with fixed particle-number. Yet, because particle number is not fixed in the grand canonical partition function, which accesses the multi-particle Fock-space, it is natural that the path integral and the grand canonical partition function be equal. In what follows, we use $Z(\tau)$ and $Z_{\rm GC}(\tau)$ interchangeably. Now for the main Theorem of this section:

\begin{thm}\label{ThmTrex}
Suppose $Z_{\rm GC}(\tau) \in F_k$. The L-function for $\langle T(\tau) \rangle = q \partial_q \log Z_{\rm GC}(\tau)$ is,
\begin{align}
L_{\langle T \rangle}^\reg(s) = -\frac{(2 \pi)^s}{\G(s)} \bigg[ \Delta\bigg(\frac{B^s}{s}\bigg) - \frac{k}{2 \pi} \bigg( \frac{B^{s-1}}{s-1} \bigg) + \Delta \bigg(\frac{B^{s-2}}{s-2}\bigg)
+ {\rm regular}(s) \bigg]~,
\label{eqThmTrex1}
\end{align}
where $-\Delta$ is the leading power of $q$ in $Z_{\rm GC}(\tau)$, and ``${\rm regular}(s)$'' refers to the terms in $L_f^*(s)^\reg$ that are finite for finite $s \in \C$. 
\end{thm}

\begin{proof}
Our proof rests on five results. First, that the weighted sum of the orders of zeros and poles of a modular form $f \in F_k$ in the fundamental domain, $\cal{F}$, is exactly given by $k$:
\begin{align}
\frac{k}{12} = \Delta + \sum_{w \in \cal{F}} e(w) ~ {\rm ord}_f(w)~. \label{eqThmTrex2}
\end{align}
Here ${\rm ord}_f(w)$ is the order of the pole or zero of $f$ at $w$, and $e(w) = 1$ unless $w = (-1)^{1/2}$ or $(-1)^{1/3}$. At these special values, $e((-1)^{1/2}) = 1/2$ and $e((-1)^{1/3}) = 1/3$. Note that $\Delta$ is the order of the zero or pole at the cusp $\tau = i \infty$. 

Second, we use Theorem 5 of~\cite{22-DivisorsMF2}. This Theorem states that for $f \in F_k$ with a $q$-series $q^{-\Delta} \sum_{n} c(n) q^n$ with $c(1) = 1$, and $\La_w(\tau) = \tfrac{1}{2 \pi i } \partial_{\tau} \log (j(\tau)-j(w))$, then
\begin{align}
q \frac{d}{dq} \log f(\tau) = \frac{k}{12} E_2(\tau) + 
\sum_{w \in \cal{F}}e(w)  {\rm ord}_f(w) \La_w(\tau)~~.
\label{eqThmTrex3}
\end{align}
Third, we use the fact that the L-function of the quasi-modular $E_2(\tau)$ is,
\begin{align}
L_{E_2}^\reg(s) = \frac{(2 \pi)^s}{\G(s)} \bigg[ \frac{6}{\pi} \frac{B^{s-1}}{s-1}-\frac{B^s}{s} - \frac{B^{s-2}}{s-2} + \sum_{n = 1}^{\infty} \sigma_1(n) \bigg( \frac{\G(s,2\pi n)}{(2 \pi n)^s} - \frac{\G(2-s,2\pi n)}{(2 \pi n)^{2-s}}\bigg) \bigg] ~. \label{eqThmTrex4}
\end{align}
Fourth, we use the fact that the L-function of the exactly modular, but meromorphic, modular form $\La_w(\tau)$ is,
\begin{align}
L_{\La_w}^\reg(s) = \frac{(2 \pi)^s}{\G(s)} \bigg[ \frac{B^s}{s} + \frac{B^{s-2}}{s-2} + {\rm regular}(s)~,\bigg] \label{eqThmTrex5}
\end{align}
where ``${\rm regular}(s)$'' refers to the terms in $L_f^*(s)^\reg$ that are finite for finite $s \in \C$.

Fifth, we use the fact that the L-function of a finite sum of modular forms is the sum of the finite number of L-functions of each of the individual modular forms. This follows from the fact that L-functions are fundamentally integral transforms of the modular forms and from the fact that an integral is a linear functional of its arguments.

Now, we use these five facts for the Borcherds products that define $Z_{\rm GC}(\tau)$. Thus, using Eq.~\eqref{eqThmTrex3} we find that
\begin{align}
L_{\langle T \rangle}^\reg(s) = \frac{k}{12} L_{E_2}^\reg(s) + \sum_{w \in \cal{F}}e(w) ~ {\rm ord}_f(w) L_{\La_w}^\reg(s)~. \label{eqThmTrex6}
\end{align}
Now, using Eqs.~\eqref{eqThmTrex4} and~\eqref{eqThmTrex5}, we have
\begin{align}
\!\!\!\!\!\!
L_{\langle T \rangle}^\reg(s) = \frac{(2 \pi)^s}{\G(s)}\bigg[ \bigg( \frac{B^s}{s}+ \frac{B^{s-2}}{s-2} \bigg) \bigg\{ \sum_{w \in \cal{F}} e(w) ~ {\rm ord}_f(w)-\frac{k}{12}  \bigg\} - \frac{k}{12} \frac{6}{\pi} \frac{B^{s-1}}{s-1} + {\rm regular}(s) \bigg] \!. \!\!\!\!\!\!
\end{align}
By Eq.~\eqref{eqThmTrex2}, we see that the common coefficient of the $1/s$ and $1/(s-2)$ poles is simply $\Delta$. This concludes the proof.
\end{proof}

Theorem~\ref{ThmTrex} concerns special values of $L_{\langle T \rangle}^\reg(s)$. Before stating a very useful Corollary, we make the following observation. When written as a formal Dirichlet series, this L-function would take the form,
\begin{align}
L_{\langle T \rangle}^\reg(s) = -\reg \sum_{n = 1}^{\infty} \bigg( \sum_{m|n} m~d(m) \bigg) n^{-s}~.
\end{align}
Now, formally this resembles a Dirichlet convolution of two Dirichlet series,
\begin{align}
\reg \sum_{n = 1}^{\infty} \bigg( \sum_{m|n} m~d(m) \bigg) n^{-s} = 
\bigg( \reg \sum_{n = 1}^{\infty} \frac{1}{n^s} \bigg) 
\bigg( \reg \sum_{m = 1}^{\infty} \frac{d(m)}{m^{s-1}} \bigg)~.
\end{align}
Note that $\sum_m d(m)/m^{s-1}$ is {\em exactly} what we would write-down for the formal Mellin transform of $Z_{\rm CAN}(\tau)$, evaluated at the shifted value $s-1$. Further, $\zeta(s)$ comes precisely from the combinatoric map between the grand canonical and canonical partition functions. On these physical grounds, it is natural to define the formal L-function for $Z_{\rm CAN}(\tau)$ as
\begin{align}
L_{Z_{\rm CAN}}^{\reg}(s):= -\frac{L_{\langle T \rangle}^\reg(s+1)}{\zeta(s+1)} ~. \label{eqLcan1}
\end{align}
Recall that in section~\ref{svTSM2}, we promised a physical motivation for the formal factorization of the L-function for $\Lambda_d(\tau)$ into $\zeta(s)$ multiplied by $\sum_n A(n^2,d)/n^{s-1}$. To give it, we draw an analogy between ${\cal H}_d(j(\tau))$ and the grand canonical partition function of a CFT, $\sum_n A(n^2,d) q^n$ and the single-particle partition function of the CFT, and $\Lambda_d(\tau)$ as the stress-energy tensor of the CFT. This analogy provides the physical motivation for the sum-rules in Eq.~\eqref{eqSumFull}. 

The sum-rules Eq.~\eqref{eqSumFull} are special cases of the following quite general Corollary to Theorem~\ref{ThmTrex}:

\begin{cor}\label{CorTrex}
If a 2d CFT path integral $Z(\tau)$ can be written as $q^{-\Delta} \prod_n (1-q^n)^{-d(n)} \in F_k$ and $Z_{\rm CAN} = \sum_n d(n) q^n$, then
\begin{align}
\begin{cases}
\lim_{s \to -1} L_{Z_{\rm CAN}}^\reg(s) 
&= \lim_{s \to -1} \reg \sum_{n = 1}^{\infty} d(n)~n^{-s} 
= +2 \Delta~,  \\
\lim_{~s \to 0~} L_{Z_{\rm CAN}}^\reg(s) 
&= \lim_{~s \to 0~} \reg \sum_{n = 1}^{\infty} d(n)~n^{-s} 
= k ~. \label{eqTrex12}
\end{cases}
\end{align}
\end{cor}

\begin{proof}
This immediately follows from Theorem~\ref{ThmTrex} and from the formal definition of the L-function for canonical ensemble in Eq.~\eqref{eqLcan1}: $L_{Z_{\rm CAN}}^\reg(s) = -L_{\langle T \rangle}^\reg(s+1)/\zeta(s+1)$.
\end{proof}

\subsection{The L-function for the stress-tensor and T-reflection sum-rules}\label{secTrex}

The two general sum-rules in Corollary~\ref{CorTrex} were anticipated in Refs.~\cite{14T-rex0, 06T-rex1, 07T-rex2, 30T-rex2-GL2}, where it was noticed that many finite-temperature path integrals in QFT were invariant under formally reflecting temperatures to negative values (T-reflection). Particularly in Refs.~\cite{06T-rex1, 07T-rex2} it was argued that if a 2d CFT path integral $Z(\tau)$ had well-defined modular weight $k \in \tfrac12 \Z$, and could be written in terms of an infinite product $q^{-\Delta} \prod_n (1-q^n)^{-d(n)}$ with $d(n) \in \Z$, then T-reflection invariance would imply two the sum-rules:
\begin{align}
q^{-\Delta} \prod_{n = 1}^{\infty} (1-q^n)^{-d(n)} \in F_k \implies
\begin{cases}
\reg \sum_{n = 1}^{\infty} n^1~d(n) &= 2 \Delta~, \\
\reg \sum_{n = 1}^{\infty} n^0~d(n) &= k ~.
\end{cases}
\end{align}
We now note that the L-functions defined in section~\ref{secLfunction} provide a concrete context in which we may evaluate these sum-rules. Theorem~\ref{ThmTrex} and Corollary~\ref{CorTrex}, validate these sum-rules when $Z(\tau)$ has zeros or poles away from cusps. (As we show in Appendix~\ref{secEg1} more conventional L-function technology verifies these sum-rules when $Z(\tau)$ has divisors only at the zero-/infinite-temperature cusp.)

As emphasized above, the stress-energy tensor plays an absolutely crucial role in QFT, and in particular in CFT. By encoding a CFTs central charge, the stress-tensor both counts the number of degrees of freedom in the system, and, more importantly, describes how a CFT responds to deformations of its spacetime manifold (i.e.~variations of $\tau$, the shape of the torus). This latter fact is of chief interest to us, here: The response to the CFT to deformations of its spacetime manifold constitutes a sort of anomalous breaking of conformal invariance: If a theory has a nonzero central charge, it has a {\em conformal anomaly}. 

Similarly, in~\cite{07T-rex2} we argue that 2d CFT path integrals are naturally both invariant under modular transformations {\em and} under T-reflection. Throughout~\cite{06T-rex1, 07T-rex2}, we argued that if a path integral were invariant under T-reflection up to such an overall phase, this phase would constitute a {\em global gravitational anomaly}~\cite{31-Witten84, 32-HeteroticII, 33-Polchinski-II}. A main role of the stress-energy tensor is to encode the anomalies of a CFT under diffeomorphisms or variations in the spacetime manifold. The new fact in this section, that the L-function of the stress-energy tensor encodes the anomaly of a CFTs path integral under the T-reflection redundancy in how the two-torus is encoded in the path integral, is thus striking. 

Not only striking, this new fact is consistent with recent work on T-reflection~\cite{06T-rex1, 07T-rex2} and related work on modular forms~\cite{30T-rex2-GL2}. In~\cite{07T-rex2}, we argue that the T-reflection phase is tied directly to the modular weight of a 2d CFT path integral: $e^{i \g} = (-1)^k$. The sum-rule in Corollary~\ref{CorTrex} explicitly verifies that this sum-rule is consistent with our L-functions for meromorphic modular forms. Further, in~\cite{30T-rex2-GL2} we explicitly construct an extension of $\SL_2(\Z)$ modular forms defined on the upper half-plane to $\GL_2(\Z)$ modular forms defined on the double half-plane, where the T-reflection phase is again given by $(-1)^k$. This explicitly agrees with the sum-rules in Corollary~\ref{CorTrex}, and fits well with the physical and mathematical structure of Refs.~\cite{06T-rex1, 07T-rex2, 30T-rex2-GL2}. In this paper, we have shown that the relevant sum-rule that counts the T-reflection phase comes exactly from the L-function for the stress-energy tensor.  

\subsection{Borcherds products and the statistical mechanical ensembles}\label{secEnsemble}

In this brief section, we note that the objects extracted from Borcherds products that we called $Z_{\rm CAN}(\tau)$ may be of independent mathematical interest. We have argued that a Borcherds product resembles the partition function for a 2d CFT on a two-torus in the grand canonical/multi-particle ensemble. The logarithm of a Borcherds product resembles the associated partition function for the same 2d CFT in the canonical/single-particle ensemble. 

This physical picture suggests that the canonical partition functions $Z_{\rm CAN}(\tau)$,
\begin{align}
Z_{\rm GC}(\tau) 
= q^{-\Delta} \prod_{n = 1}^{\infty}  (1-q^n)^{-d(n)}
\quad \mapsto \quad 
Z_{\rm CAN}(\tau) 
= d(0) + \sum_{n = 1}^{\infty} d(n) q^n 
~, \label{eqZmap1}
\end{align}
may be of independent mathematical interest. Certainly, for the simplest example,
\begin{align}
Z_{\rm GC}(\tau) = \frac{1}{q^{1/24}} \prod_{n = 1}^{\infty} \frac{1}{(1-q^n)} = \frac{1}{\eta(\tau)}
\quad \mapsto \quad 
Z_{\rm CAN}(\tau) = \frac{1}{1-q} 
~, \label{eqZmap0}
\end{align}
the partition function in the canonical ensemble is a character of $\SL_2(\R)$. 

Juxtaposing the Borcherds products in section~\ref{svTSM2} and in~\cite{02-TSM} against the physical picture where a modular form is some ``grand canonical partition function'' for a free statistical system, suggests that the ``canonical partition functions''  associated to ${\cal H}_d(j(\tau))$,
\begin{align}
Z_{\rm GC}^{(d)}(\tau):= q^{-H(d)} \prod_{n = 1}^{\infty} (1-q^n)^{A(n^2,d)} = {\cal H}_d(j(\tau)) 
\quad \mapsto \quad
Z_{\rm CAN}^{(d)}(\tau) := \sum_{n = 1}^{\infty} A(n^2,d) q^n ~,
\end{align}
are mathematically interesting objects. However, we are unaware of any study of these functions in the literature. It might be pleasing if the other functions for $d > 0$ also have a clear meaning in terms of characters of non-compact groups, or played a more significant role in the study of e.g. traces of singular moduli. We leave such questions to future work.

\section{Summary and future directions}\label{secEnd}

{\em Context for the paper:} The motivation for this paper was to verify a conjecture of T-reflection~\cite{14T-rex0}, in the physics-agnostic setting of the mathematics of modular forms~\cite{06T-rex1, 07T-rex2}. To do this, we extended the recent map between weakly holomorphic modular forms~\cite{03-BFI, 04-BFK, 05-BFK2} and regularized L-functions/regularized Dirichlet series -- itself, an extension of the classic map between holomorphic modular forms and L-functions/Dirichlet series~\cite{01-Apostol} -- to a map between meromorphic modular forms and regularized L-functions/regularized Dirichlet series.

{\em Summary and main results of the paper:} We have developed a map between meromorphic modular forms and regularized L-functions, presented in Theorem~\ref{ThmLfn}. To do so, we had to explicitly and precisely understand the exponential growth in $q$-series coefficients, presented in Proposition~\ref{PropGrowthCap}. We then used these results to understand aspects of traces of singular moduli, when viewed as $q$-series coefficients of meromorphic modular forms. Many of these aspects hinge critically on L-functions for meromorphic modular forms, and thus are new.

In a seemingly different direction, we viewed meromorphic modular forms as models for CFT path integrals and partition functions in various ensembles in statistical mechanics. Concretely, we mapped meromorphic modular forms onto partition functions for the full multi-particle Fock-space of free 2d CFTs, and we mapped (roughly) the logarithm of meromorphic modular forms onto partition functions for the single-particle Fock-space of a the same free 2d CFTs. We then studied the L-function of the stress-energy tensor, which is the logarithmic derivative of meromorphic modular partition function, of these CFTs, and studied its special values. The stress-tensor is a weight-two meromorphic quasimodular form. Its regularized L-function has special values at $s = 0$, $s = 1$ and $s = 2$. 

In Theorem~\ref{ThmTrex} and Corollary~\ref{CorTrex}, we explicitly showed that the $s = 0$ special values of the stress-energy tensor L-function allow us to equate the central charge of a free and unitary 2d CFT with a regularized tally of the total number of states in the single-particle Fock-space. Translated back into the mathematical context of traces of singular moduli in, this is equivalent to the statement that the regularized sum of traces of singular moduli for the family of modular functions $J_n(\tau) := q^{-m} + {\cal O}(q) \in M_0^!$, when evaluated at CM points of discriminant-$d$ quadratics, equals (minus) the Hurwitz class number $H(d)$: $\reg \sum_n t_n(d) = -H(d)$. Though these statements do not have obvious practical utility, explicitly relating the central charge to the (in general) exponentially divergent sum over states in a CFT and explicitly relating the exponentially divergent sum of the mathematically interesting $t_n(d)$ to the Hurwitz class numbers, are both new and pleasing results. Further, the $s = 1$ special value of the stress-energy tensor L-function verifies the sum-rule for the T-reflection phase, conjectured in~\cite{06T-rex1, 07T-rex2}, and hints at important physical aspects of this newly found symmetry.

{\em Bugs, features, and the future:} The motivation for this paper is solidly from theoretical physics. The main advances in this paper are largely mathematical in nature. The interplay between these two is somewhere in between. Bugs, features, and future of the results in this paper, and the lines of reasoning that led to it are most easily discussed in the three following brief sections: one on completely mathematical aspects, one on the interface between mathematics and theoretical physics that spurred this note, and one on physical aspects.

\subsection{Mathematics: comments, room for improvement, and exploration}

In this section, we describe various different routes for improvement in our treatment of L-functions of meromorphic modular forms, make several comments that did not fit into the narrative of sections~\ref{secLfunction} and~\ref{secSV}, and emphasize directions for future exploration. 

{\em L-integrals when $f \in F_k$ has poles at cusps:} As emphasized in section~\ref{secLfunction}, in order to make the L-integral for meromorphic modular forms well-defined, we made use of contour deformations. It is important that we regulated the integral while keeping the integrand meromorphic. However, as discussed in~\cite{04-BFK, 05-BFK2}, when $f \in M_k^!$ has poles at cusps, it is natural to regularize the L-integral by deforming the integrand by the non-holomorphic factor $e^{2 \pi \up {\rm Im}(\tau)} (f(\tau) -c_f(0))$. When ${\rm Re}(\up)$ is larger than the order of the pole at the cusp, then the $\up$-deformed L-integral converges, and can be continued to $\up = 0$. This procedure yields the terms Eq.~\eqref{eqLemReg2}. In this sense, our results agree with and generalize the results for $L_f^*(s)^\reg$ in the literature. However, when $f \in F_k$ has a pole at cusps then there are a finite number of terms in the next line~\eqref{eqHouston}, which diverge. 

It is not clear (to me) how to reconcile the tension between non-holomorphic deformations of the integrand from~\cite{04-BFK, 05-BFK2} to regularize poles at cusps while also being able to use residue theorems and contour-deformations to regularize divergences from poles away from cusps. I will list some possible routes. First, it is possible that there is a sense in which the $\La$-regulated L-integral is given by the constant term in $L_f^*(s|\La,\e)^\reg$, when expanded for large $\La$. This would excise the finite number of problematic (divergent) terms in Eq.~\eqref{eqHouston}. Second, it would be interesting to consider a hybrid regulator, where we look at
\begin{align}
L_f^*(s|\La,\e,\A):= \int_{\g(\La,\e)} \frac{d\tau}{\tau} \left(\frac{\tau}{i}\right)^s J(\tau)^{\alpha} (f(\tau) -c_f(0))~,
\end{align}
where $J(\tau) = q^{-1} + \sum_{n \geq 1} c(n) q^n \in M_0^!$, $J(\tau)^\A = e^{\A \log J(\tau)}$ and $\log x$ is defined to be in the principal branch: $-\pi < {\rm Im}(\log z) \leq \pi$. With this deformation, the integrand should be meromorphic throughout, allowing straightforward use of residue theorems to $\e$-deform the contour away from poles away from cusps. For sufficiently large ${\rm Re}(\A)\gg1$ the poles at cusps are also regularized. This could in principle allow one to define a regularized integral of a meromorphic integrand that has a smooth limit as $\La \to \infty$ and $\e, \A \to 0$. As our chief interest was in meromorphic modular forms that are regular at cusps, we leave this question for future work.

{\em Poles, path-dependence, and wall-crossing:} Meromorphic functions have poles. To define their L-functions, we must take a contour integral. Insisting the L-function integrand be meromorphic in $\tau$, rather than e.g. real-analytic function of $t = {\rm Im}(\tau)$, then it is natural to consider what happens when a pole smoothly moves from $z$ to $\g z$, when this path crosses the contour. If $\g \in \SL_2(\Z)$, then the locations of set of poles at $\{\g z \mid \g \in \SL_2(\Z)\}$ is mapped back to itself. As the integrand has poles at all modular images of $z$, it is invariant under {\em discrete} modular transformations $z \to \g z$. However, the integral is not invariant if we {\em continuously} deform $z$ to $\g z$ if the path crosses the original contour along the imaginary-$\tau$ axis. When this happens, then the L-function picks-up residues. 

There seems to be tension here, analogous to the tension between modularity of the theta-decomposition of meromorphic Jacobi forms and modularity in Zwegers work~\cite{23-Zwegers} and others~\cite{24-DMZ}. There, as here, the non-invariance of the integral seems to be related to path-dependence, which we discussed in section~\ref{sec15-HR1919}. It would be very interesting to understand this phenomena in more detail. (Note: these residues do {\em not} alter the residues of $\tfrac1s$ or $\tfrac{i^k}{k-s}$ in $L_f^*(s)^\reg$. Special values in sections~\ref{secSV} and~\ref{secQFT} are thus unaffected by this ambiguity.)

{\em Polar Rademacher sums and numerical evaluation of $L_f^*(s)^\reg$:} It would be very interesting to systematically isolate the contribution of each individual pole to the $q$-series coefficients for $f \in F_k$ when $k > 0$, in a Rademacher-like sum. Hardy and Ramanujan did this for $1/E_6(\tau) \in F_{-6}$~\cite{15-HR1919}. Recent, beautiful, work extends this to wide classes of meromorphic modular forms $f \in F_{k}$ with negative weight~\cite{16-BialekPol1, 17-BialekPol2, 18-BialekPol3, 19-KolnPol1, 20-KolnPol2}. 

However, expressions for $q$-series coefficients for meromorphic forms with $k > 0$ seem not to be written in terms of the polylogarithms needed to explicitly relate exponential $q$-series growth due to explicit pole locations in $\HH$. If such an extension existed, by simply deleting exponential contributions to $c_f(n)$ from poles above the line ${\rm Im}(\tau) = B$, we could directly obtain rapidly convergent expressions for $c_{Rf}(n|B)$ used in $L_f^*(s)^\reg$ for any $f \in F_k$. (Note: $q$-series coefficients for some $f \in F_k$ for $k \geq 0$ in terms of exponential sums for $k > 0$ are known, e.g.~\cite{34-Duke-pole, 35-KolnPol3}. Yet, the map between terms in the exponential sums and pole locations, which follows naturally from Rademachers sums, seems less direct here.) 

{\em The constants $c_f(0)$ decouple from L-functions for $f \in F_0$:} 
The integral that defines $L_f^*(s)^\reg$ is $\int_0^{\infty} dt~t^{s-1}~(f(it)-c_f(0))$. So when $f \in F_0$, then all L-functions are just functions of the non-constant $q$-series coefficients: the $c_f(0)$ drops out. Amusingly, $c_f(0)$ is also an element of $F_0$. This may suggest a fruitful reformulation of L-functions for $f \in F_k$. Recall that when $f \in M_k$, we subtract-off the modular function called $c_f(0)$ from $f(\tau)$ within $L_f^*(s)$. This has the effect of subtracting-out the pole in Mellin-space that comes from a residue at $\tau = i \infty$. When $f \in M_k^!$ or $f \in F_k$, rather than subtracting-off the modular function called $c_f(0)$ which cancels the pole coming from $t \to 0$ and $t \to \infty$, we could subtract-off a modular function that cancels all poles along the integration contour but introduces no new poles. In this scenario, the integrand would be meromorphic in $\tau$ throughout, and subtleties of $L_f^*(s)$ would be contained within $L_f^\reg(s)$ for $f \in F_0$ (and $f \in M_0^!$).

{\em For future exploration:} 
It would be good to fit L-functions for $f \in F_k$ (and $f \in M_k^!$~\cite{04-BFK, 05-BFK2}) 
into the web of conjectures and facts about L-functions for $f \in M_k$ and $S_k$. It is natural to extend $L_f^\reg(s)$ for $f \in F_k(\G)$ to subgroups $\G < \SL_2(\Z)$, or to half-integral $k \in \tfrac12 \Z$. L-functions for weakly holomorphic modular forms with Zagier duality particularly stand-out. 

\subsection{Field theory, statistical mechanics, Hagedorn, and number theory}

There is a strong interrelation between the mathematical and physical perspectives in this paper. This goes in both directions.

{\em Borcherds products as grand canonical partition functions:} In section~\ref{secQFT}, we emphasized that if a modular form has an infinite product of the form $F(\tau):= q^{-\Delta} \prod_n (1-q^n)^{-d(n)} \in F_k$, then it can be interpreted in some sense as the partition function for the multi-particle Fock-space of a free CFT, i.e. a grand canonical partition function. In this guise, then, one can define the partition function for the associated single-particle Fock-space of the free CFT. It is given by $f(\tau):= \sum_n d(n) q^n$. (A combinatoric map relates $f$ and $F$. See section~\ref{secSTL}.) 

Physically, both $F$ and $d$ play very important roles. However, we have been unable to find any discussion of $f(\tau)$ in the mathematics literature, other than for the simple case of $F(\tau) = 1/\eta(\tau)$, where $f(\tau) = 1/(1-q)$ is a character of $\SL_2(\R)$. It would be very interesting to see if these ``single particle partition functions'' $f(\tau)$ played a comparably important role in number theory---particularly when $F$ has a Borcherds product expansion, and $f(\tau) = -\sum_n A(n^2,d) q^n$ has such a close relationship with traces of singular moduli. 

{\em Strings, modularity and Hagedorn:} As emphasized in the Introduction, both the low-energy limit of QCD and the field theory limit of string theories generically exhibit an exponential/Hagedorn rise in the number of states with energy, $d(E)$:
\begin{align}
d(E) \sim E^{\alpha} e^{\beta_H E}~.
\end{align} 
When the inverse-temperature $\beta$ approaches $\beta_H$, then the one-loop path integral/partition function $Z(\beta) = \sum_E d(E) e^{-E \beta}$ diverges. Additionally, loops of closed strings, either QCD-strings or more fundamental strings, are topologically equivalent to a torus. Thus, we expect the one-loop path integral in these theories, $Z(\beta)$, to be modular in the $\beta$-parameter. 

Putting these two features together naturally suggests that the theory of meromorphic modular forms and perhaps the related theory of meromorphic Jacobi forms~\cite{24-DMZ} may play a role in understanding the path integral and observables in these physical contexts. In this paper, we studied whether one can define a regularization for the exponentially divergent sums that appear in Casimir energies in these field theories. Ubiquity of modularity and Hagedorn growth/poles suggests L-functions for meromorphic modular forms may directly yield Casimir energies of models of low-energy QCD. It would be very interesting if other aspects of meromorphic modular or Jacobi forms interplay with low-energy QCD. 

There is at least one explicit precedent. Namely, recently, it was realized that the path integral of a famously tractable limit of QCD in four-dimensions~\cite{09-S1S3} is given by~\cite{10-2d4d1, 11-2d4d2}
\begin{align}
Z(\beta) = \prod_{n =1 }^{\infty} \frac{(1-q^n)}{(1+q^n)(1-z q^n) (1-q^n/z)} ~~,~~z = 2 + \sqrt{3}~~.
\end{align}
This can be easily recognized as the quotient of Jacobi theta-functions and Dedekind eta-functions. It has simple poles when $q^n \to z^{\pm 1}$. It would be very interesting to understand if, e.g. the wall-crossing in~\cite{24-DMZ} had some role in models of low-energy QCD, for instance in this particular path integral that is meromorphic and involves Jacobi theta functions.

{\em Fermionic symmetries:}
Finally, we make an amusing standalone observation. Let $F(\tau)= q^{-\Delta} \prod_n (1-q^n)^{-d(n)} =q^{-\Delta} \sum_n D(n) q^n$ be a weakly holomorphic modular form with zeros at finite values of $\tau$ and poles at cusps. Its $q$-series coefficients are bounded by $|D(n)| < e^{C \sqrt{n}}$ for some $C>0$. Because $F(z) = 0$ for some $0< |z| < \infty$, then Proposition~\ref{PropGrowthCap} implies $|d(n)| > E (e^{2 \pi {\rm Im}(z)})^n/n$ for some $E>0$ for infinitely many positive integers $n$. 

Thus, there are huge cancellations between the $q$-series coefficients of $f(\tau) = \sum_n d(n) q^n$ and $F(\tau) = q^{-\Delta} \sum_n D(n) q^n$. Very similar cancellations were observed in~\cite{36-Fermionic1, 37-Fermionic2, 38-Fermionic3, 39-Fermionic4} from a fermionic symmetry in non-supersymmetric models of QCD with Hagedorn growth in the number of states at a given energy $E$. Despite the lack of supersymmetry and the presence of Hagedorn growth, the bosons and fermions are almost exactly paired and cancel in the twisted path integral, which avoids Hagedorn poles at finite temperature.

\subsection{Central charges, T-reflection phases, and the stress-tensor}

Finally, in sections~\ref{secSTL} and~\ref{secTrex} we showed that both the Casimir energy and the T-reflection sum-rule of a 2d CFT are captured by the L-function of the stress-energy tensor, $L_{\langle T \rangle}^{\reg}(s)$, when evaluated respectively at $s = 0$ and $s = 1$. Central charges give the conformal anomaly of the CFT when it is on curved manifolds. So it is natural that they are captured by the stress-tensor. However, it was completely unexpected that the T-reflection sum-rules suggested in~\cite{14T-rex0, 06T-rex1} would literally appear next to the central charge, in the stress-energy tensor L-function, at $s = 0$ and $1$. Yet, this may be sensible in light of Refs.~\cite{06T-rex1, 07T-rex2}, where we noted that the T-reflection phase has a natural interpretation as a global gravitational anomaly. It seems existentially important to better understand the connection between the T-reflection phases and the regularized L-function for the stress-energy tensor of 2d CFTs.

\section*{Acknowledgements}

This work was supported by the Niels Bohr International Academy (NBIA), and by a Carlsberg Distinguished Postdoctoral Fellowship (CF16-0183) at the NBIA. I would like to thank Kathrin Bringmann for the opportunity to present this work prior to posting, and for correspondence. I would like to thank Kathrin Bringmann, Ben Kane, Michael Mertens, Ken Ono and Sander Zwegers for related discussions. 
%Also Apollo, my lovely sister's dog from Benin. Apollo is great!
I would also like to thank iNes Aniceto for collaboration during early stages of this work. Her input gave extremely important data in support of the sum-rules in Theorem~\ref{Thm2} and Corollary~\ref{Cor1}. Finally, I would like to thank John Duncan for extensive discussions and collaboration at early stages of this project and collaboration on the related project~\cite{30T-rex2-GL2}.

\appendix

\section{T-reflection sum-rules when $\langle T(\tau) \rangle$ is holomorphic and quasimodular}\label{secEg1}

Here, we motivate the T-reflection sum-rules for the special case of a free scalar CFT where $Z(\tau) = Z_{\rm GC}(\tau) = 1/\eta(\tau) \in M_{1/2}^!$, i.e. where the path integral has divisors only at the zero-/infinite-temperature cusp. This formal analysis originally appeared in~\cite{14T-rex0}, and was recently explained in greater detail in~\cite{06T-rex1}. 

To start, we show that $1/\eta(\tau)$ corresponds to the path integral for a single free scalar CFT on the two-torus. If we think of the rectangular two-torus as the direct product of a periodic line element,  $S^1_L$, with the thermal circle, $S^1_{\beta}$, then we can easily see that the single-particle Fock-space for a single scalar particle is simply that of a particle with momentum $p_n = n/L$ with $n = 0, 1, 2, \ldots \in \Z_{\geq 0}$. Further, we define $2 \pi i \tau:= -\beta/L$ and $q := e^{-\beta/L} = e^{2 \pi i \tau}$.

As there is only one scalar, and only one direction in which to move, we find that 
\begin{align}
d(n) = 1~{\rm for}~n = 0, 1, 2 \ldots~\implies~Z_{\rm CAN}(\tau) = \sum_{n = 0}^{\infty} d(n) q^n = \frac{1}{1-q} .
\end{align}
If we apply the combinatoric map between the single-particle and multi-particle Fock-spaces for this free CFT, we find that $Z_{\rm GC}(\tau)$ evaluates to 
\begin{align}
Z_{\rm GC}(\tau) = q^{-\Delta} \prod_{n = 1}^{\infty} \frac{1}{1-q^n}~.
\end{align}
Now, if we take the logarithmic derivative of $Z_{\rm GC}(\tau)$ to find the one-point correlation function of this 2d CFTs stress-energy tensor, we find it equals the following expression:
\begin{align}
\langle T(\tau) \rangle = \frac{1}{2 \pi i} \frac{d}{d\tau} \log \bigg( q^{-\Delta} \prod_{n = 1}^{\infty} \frac{1}{1-q^n}\bigg) = -\Delta - \sum_{n = 1}^{\infty} q^n \bigg( \sum_{m|n} m \bigg) ~.
\end{align}
It is straightforward to compute that the Mellin transform, i.e. the L-function, for the one-point correlation function for the stress-energy tensor of this 2d CFT is
\begin{align}
L_{\langle T \rangle}(s) = \frac{(2 \pi)^s}{\G(s)} \int_0^{\infty} \frac{dt}{t} t^s ( \langle T(it) \rangle + \Delta) = -\zeta(s) \zeta(s-1)~.
\end{align}
Because $\langle T(\tau) \rangle$ neither diverges nor vanishes at any finite value of $\tau = \beta/(2 \pi i L)$, the regularization procedure developed in this paper is not necessary to define $L_f(s)$. 

We can also directly compute the L-function of $Z_{\rm CAN}(\tau)$, and find that it equals
\begin{align}
L_{Z_{\rm CAN}}(s) = \frac{(2\pi)^s}{\G(s)} \int_0^{\infty} \frac{dt}{t} t^s \sum_{n = 1}^{\infty} e^{-2 \pi n t} = \zeta(s)~.
\end{align}
It is important to note that the formal factorization in Eq.~\eqref{eqLcan1} would imply that $L_{Z_{\rm CAN}}(s) := -L_{\langle T \rangle}(s+1)/\zeta(s+1)$. Crucially, these computations match. (Note that $\langle T (\tau) \rangle = -E_2(\tau) + {\cal O}(q^0)$, and is tightly related to holomorphic quasimodular forms. See Eq.~\eqref{eqThmTrex4} for an explicit expression for the L-function of $E_2(\tau)$.)

Finally, we would like to formally motivate the T-reflection sum-rules in section~\ref{secTrex} for this special case. First, note that $Z_{\rm GC}(\tau)$ resembles a product of decoupled products of harmonic oscillator partition functions, $1/(1-q^n)$, stripped of their zero-point energies $q^{n/2}$:
\begin{align}
\prod_n Z_{\rm osc.}^{(n)}(\tau) := \prod_n \frac{q^{n/2}}{1-q^n}~.
\end{align}
We may recover the grand canonical partition function by taking a large (infinite) collection of decoupled oscillators that have been properly endowed with their zero-point energies $q^{n/2}$. Now, note that if we regulate the sum of exponents $\sum_n n/2$ by the L-function for the canonical ensemble, i.e. $\zeta(s)$, then we would find
\begin{align}
\lim_{N \to \infty} \Reg \prod_{n = 1}^{N} Z_{\rm osc.}^{(n)}(\tau) = \lim_{N \to \infty} \Reg \prod_{n = 1}^{N} \frac{q^{n/2}}{1-q^n} = q^{\frac{1}{2} \zeta(-1)} \prod_{n = 1}^{\infty} \frac{1}{1-q^n} = \frac{1}{\eta(\tau)}~.
\end{align}
Further, note that each decoupled oscillator is invariant under T-reflection~\cite{14T-rex0, 06T-rex1}. (This follows from the fact that $q^{n/2}/(1-q^{n})$ equals $\sin(\pi n \tau)^{-1}$, an odd function of $\tau$.) Thus:
\begin{align}
\prod_n Z_{\rm osc.}^{(n)}(\beta) = \prod_n \frac{q^{n/2}}{1-q^n} \to 
\prod_n Z_{\rm osc.}^{(n)}(-\tau)= \prod_n \frac{q^{-n/2}}{1-q^{-n}} = \prod_n \frac{(-1)~q^{n/2}}{1-q^n}~.
\end{align}
Again, we take the limit where the number of oscillators goes to infinity, which is needed in order for the states accessed in the partition function to span the full multi-particle Fock-space, {\em but now at negative temperature}. Here, we must regulate the divergent product of zero-point energies {\em and} the divergent product of $(-1)$-factors. Regulating with $\zeta(s)$, we find:
\begin{align}
\!\!\!\!\!\!
\lim_{N \to \infty} \Reg \prod_{n = 1}^{N} Z_{\rm osc.}^{(n)}(-\tau) = \lim_{N \to \infty} \Reg \prod_{n = 1}^{N} \frac{(-1)~q^{n/2}}{1-q^n} = q^{\frac{1}{2} \zeta(-1)} (-1)^{\zeta(0)} \prod_{n = 1}^{\infty} \frac{1}{1-q^n} = \frac{(-i)}{\eta(\tau)}~.\!\!
\end{align}
Thus, we see that the leading power of $q$ in the $q$-series of this particular path integral is fixed by T-reflection. Further, the eigenvalue under T-reflection is $e^{i \g}= -i$, consistent with the sum-rule in Corollary~\ref{CorTrex} and with Refs.~\cite{06T-rex1, 07T-rex2, 14T-rex0, 30T-rex2-GL2}, where $e^{i \g} = (-1)^k = (-1)^{\reg \sum_n d(n)}$.

\end{document}